\documentclass[journal,onecolumn,12pt, draftclsnofoot]{IEEEtran}
\IEEEoverridecommandlockouts
\usepackage{cite}
\usepackage{amssymb,amsfonts}
\usepackage{graphicx}
 
\usepackage{textcomp}
\usepackage[tbtags]{amsmath}
\usepackage{diagbox}
\usepackage{setspace}
\usepackage{booktabs}
\usepackage{multirow}
\usepackage{hyperref}
\usepackage{amsthm} 
\usepackage{algorithm}
\usepackage{algpseudocode}
\usepackage{color}
\alglanguage{pseudocode}
\algblockdefx{Process}{EndProcess}[1]{{\bf process} #1 }{{\bf end process}}
\algblockdefx{Parallel}{EndParallel}[1]{{\bf begin parallel processing} #1 }{{\bf end parallel processing}}
\algblockdefx{FP}{EndFP}[1]%
{\textbf{for all} #1 \textbf{do in parallel}}%
{\textbf{end for}}

\usepackage{accents}
\newcommand{\ubar}[1]{\underaccent{\bar}{#1}}

\usepackage{tikz,bm,adjustbox}
\usepackage[flushleft]{threeparttable}
\usepackage{soul}
\usetikzlibrary{calc}
\usetikzlibrary{decorations.pathmorphing}

\makeatletter

\newcommand{\defhighlighter}[3][]{%
	\tikzset{every highlighter/.style={fill opacity=0.3,line width=0.25em,opacity=0.3}}%
}

\defhighlighter{yellow}{.5}

\newcommand{\highlight@DoHighlight}{
	\fill [ decoration = {amplitude=1pt, segment length=15pt}
	, outer sep = -15pt, inner sep = 0pt, decorate
	, every highlighter, this highlighter ]
	($(begin highlight)+(0,8pt)$) rectangle($(end highlight)+(0,-3pt)$) ;
}

\newcommand{\highlight@BeginHighlight}{
	\coordinate (begin highlight) at (0,0.1) ;
}

\newcommand{\highlight@EndHighlight}{
	\coordinate (end highlight) at (0,-0.1) ;
}

\newdimen\highlight@previous
\newdimen\highlight@current

\DeclareRobustCommand*\highlight[1][]{%
	\tikzset{this highlighter/.style={#1}}%
	\SOUL@setup
	\def\SOUL@preamble{%
		\begin{tikzpicture}[overlay, remember picture]
		\highlight@BeginHighlight
		\highlight@EndHighlight
		\end{tikzpicture}%
	}%
	\def\SOUL@postamble{%
		\begin{tikzpicture}[overlay, remember picture]
		\highlight@EndHighlight
		\highlight@DoHighlight
		\end{tikzpicture}%
	}%
	\def\SOUL@everyhyphen{%
		\discretionary{%
			\SOUL@setkern\SOUL@hyphkern
			\SOUL@sethyphenchar
			\tikz[overlay, remember picture] \highlight@EndHighlight ;%
		}{%
		}{%
			\SOUL@setkern\SOUL@charkern
		}%
	}%
	\def\SOUL@everyexhyphen##1{%
		\SOUL@setkern\SOUL@hyphkern
		\hbox{##1}%
		\discretionary{%
			\tikz[overlay, remember picture] \highlight@EndHighlight ;%
		}{%
		}{%
			\SOUL@setkern\SOUL@charkern
		}%
	}%
	\def\SOUL@everysyllable{%
		\begin{tikzpicture}[overlay, remember picture]
		\path let \p0 = (begin highlight), \p1 = (0,0) in \pgfextra
		\global\highlight@previous=\y0
		\global\highlight@current =\y1
		\endpgfextra (0,0) ;
		\ifdim\highlight@current < \highlight@previous
		\highlight@DoHighlight
		\highlight@BeginHighlight
		\fi
		\end{tikzpicture}%
		\the\SOUL@syllable
		\tikz[overlay, remember picture] \highlight@EndHighlight ;%
	}%
	\SOUL@
}
\makeatother
\theoremstyle{Remark}
\theoremstyle{Example}

\newtheorem{Theorem}{Theorem}
\newtheorem{Example}{Example}
\newtheorem{Corollary}{Corollary}
\newtheorem{Lemma}{Lemma}
\newtheorem{Remark}{Remark}
\theoremstyle{definition}

\def\BibTeX{{\rm B\kern-.05em{\sc i\kern-.025em b}\kern-.08em
		T\kern-.1667em\lower.7ex\hbox{E}\kern-.125emX}}

\begin{document}
 	  
	\title{Coded Caching for  Broadcast Networks with User Cooperation}
	\author{ \IEEEauthorblockN{Jiahui Chen, Xiaowen You, Youlong Wu, and Shuai Ma}
  \thanks{This paper was in part presented at the \emph{IEEE Information Theory Workshop}, Visby, Gotland, Sweden,  2019   and  at \emph{57th Annual Allerton Conference on Communication, Control, and Computing (Allerton)}, Monticello, IL, USA, 2019.

}

}

%

	\maketitle
	
	\begin{abstract}
		In this paper, we investigate the transmission delay of cache-aided broadcast networks with user cooperation. Novel coded caching schemes are proposed for both centralized and decentralized caching settings, by efficiently exploiting time and cache resources and creating parallel data delivery at the server and users. We derive a lower bound on the transmission delay and show that the proposed centralized coded caching scheme is \emph{order-optimal} in the sense that it achieves a constant multiplicative gap within the lower bound. Our decentralized coded caching scheme is also order-optimal when each user's cache size is larger than the threshold $N(1-\sqrt[{K-1}]{ {1}/{(K+1)}})$ (approaching 0 as $K\to \infty$), where $K$ is the total number of users and $N$ is the size of file library. Moreover, for both the centralized and decentralized caching settings, our schemes obtain an additional \emph{cooperation gain} offered by user cooperation and an additional \emph{parallel gain} offered by the parallel transmission among the server and users. It is shown that in order to reduce the transmission delay, the number of users parallelly sending signals should be appropriately chosen according to user's cache size, and alway letting more users parallelly send information could cause high transmission delay.  	\end{abstract}
	
	
	\section{Introduction}
	
	 Caching is a promising approach that can significantly reduce traffic load in a communication network by shifting the network traffic to the low congestion periods.   
	Recently, in the seminal paper \cite{Centralized} Maddah-Ali and Niesen considered a cache-aided broadcast network where a server connects with multiple users with a shared link, and proposed  a \emph{centralized} coded caching scheme    based on   centralized file placement and coded multicast delivery.  Compared to   the conventional caching scheme, the coded caching scheme achieves a significantly larger global multicast gain. Following the similar   idea, they extended the scheme to   the decentralized file placement where no coordination is required for the file placement,  referred to as \emph{decentralized} coded caching scheme\cite{Decentralized}.   
	
	These  coded  caching schemes have attracted wide and significant interests. For the   same cache-aided broadcast network,   \cite{ImprovedGap2} showed that the rate-memory tradeoff of the above caching system is  within a factor of 2.00884 for both the peak rate and the average rate.  
	For the setting with   uncoded file placement  where each user stores  uncoded content from the library, \cite{Wan'ITW,Yu'IT18} proved that the coded caching scheme is optimal.   In \cite{PDA}, both the placement and delivery phases of coded caching are depicted using a placement delivery array (PDA), and an upper bound for all possible regular PDAs was established.   	 In \cite{HeterogeneousCacheSizes}, the authors studied  a cached-aided network with heterogeneous setting where the users's cache memories are   unequal. More asymmetric network settings have been discussed, such as coded caching with heterogeneous user profiles \cite{cHeterogeneous'19},   with distinct sizes of files \cite{filesize}, with asymmetric cache sizes \cite{cHeterogeneous'17,HeterogeneousOptimization,DistinctCapacity} and with distinct link qualities \cite{AsymmetricCacheSize'19}. The settings with varying file popularities have been  discussed in \cite{nonuniform,Randomdemands,ArbitraryDistributions,Onlinecaching}. Coded caching that accounts for various heterogeneous aspects was studied in \cite{Daniel'20}.  Other work on coded caching  include, e.g.,  cache-aided noiseless multi-server network \cite{multi-server}, cache-aided wireless/noisy broadcast network \cite{Zhang,b11,Tandon1,Tandon2}, cache-aided relay networks \cite{hierarchical,hierarchical'19,NetworkwithRelay}, cache-aided interference management \cite{intermanage'17,StoreandLate'17},  coded-caching with random demands \cite{randomdemand},     caching in combination networks \cite{combinationnetwork}, coded caching under secrecy constraints \cite{Ravindrakumar'ISIT}, coded caching with reduced subpacketization \cite{Tang'IT,Cheng'IT21}, coded caching problem where each user requests multiple files \cite{Wan'IT21}, cache-aided broadcast network for correlated content \cite{Hassanzadeh'IT21} etc.

	{A different line of work is to study the cached-aided networks without   the presence of server, e.g., the device-to-device (D2D) cache-aided network.    
	In \cite{D2D}, the authors   investigated coded caching for    wireless D2D network \cite{D2D}, where users locates in a fixed mesh topology wireless D2D network.}  A D2D system with selfish users that do not participate in delivering the missing subfiles to all users was studied in \cite{ParticialD2D}.  Wang \emph{et al}. used the  PDA  to characterize cache-aided D2D wireless networks in \cite{PDAforD2D}.  In  \cite{Malak'18} the authors studied the spatial   D2D  networks in which the D2D user locations are modelled by a Poisson point process.    For heterogeneous cache-aided D2D networks where   users are equipped with cache memories of distinct sizes,  \cite{D2DArbitraryCacheSize} minimized the delivery load by optimizing over the  partition during the placement phase and the size and structure of D2D during the delivery phase. 	Highly-dense wireless network with device mobility was investigated in  \cite{Pedersen'19}. In fact, combining  the cache-aided  broadcast network  with the cache-aided D2D network  can potentially reduce the transmission latency. This hybrid network is   common   in many practical     distributed systems such as cloud  network \cite{Fog'Overiew}, where  a central cloud server broadcasts messages to multiple users through the cellular network, and meanwhile  users  communicate with each other through a fibre local area network (LAN).  Unfortunately, there is very few work  investigating this hybrid network.

	 In this paper, we study a hybrid cache-aided network where a server consisting of $N\in\mathbb{Z}^+$ files connects with $K\in\mathbb{Z}^+$ users and meanwhile the users can communicate with each other via a cooperation network.    Unlike the settings of \cite{D2D,Malak'18} in  which  each user can only communicate with its neighbouring users via spatial multiplexing,    we consider the cooperation network as either a shared link or a flexible routing network\cite{multi-server}. 	In particular, for the case of  the shared link, all users connect with each other via a shared link. In the flexible routing network, there exists a routing strategy adaptively  partitioning all users  into multiple groups, in each of which one user sends a data packet  to the remaining users in the corresponding   group.      Let $\alpha\in \mathbb{Z}$ be the   number of groups who send signals at the same time, then there are several  
  interesting questions arising for this hybrid cache-aided network: 1) How does  $\alpha$ affect the system performance; 2) What's the (approximately) optimal value of $\alpha$ in order to minimize the transmission latency; 3) How to allocate  communication loads between the server and users,  and to design the data placement and delivery strategies  to achieve the minimum   transmission latency. 	In this paper, we try to address these questions and  our  main contributions  is summarized   as follows:

		\begin{itemize}
		\item  We {propose} coded caching schemes with user cooperation for the centralized setting and decentralized setting, respectively.  	 Both schemes efficiently exploit user cooperation and  allocate communication loads between the server and users. It is shown that our schemes can achieve much smaller transmission delay compared to the scheme without user cooperation \cite{D2D} and the scheme without server transmission \cite{Centralized,Decentralized}. We characterize  a \emph{cooperation gain}  a \emph{parallel gain} achieved by our schemes, where the cooperation gain is obtained   through  cooperation transmission between   users and the   {parallel gain} is obtained through the parallel transmission between the server and multiple users.

			
			\item  A lower bound on the transmission delay is established. With the proposed lower bound, we show that the centralized scheme achieves the optimal transmission delay within a constant multiplicative gap in all regimes, and the decentralized scheme approaches the information theoretic lower bound with a constant factor when the cache size of each user $M$ is larger than   the threshold $N(1-\sqrt[{K-1}]{ {1}/{(K+1)}})$ that is approaching to 0 as $K\to\infty$. 
			
				\item   In the  centralized  caching case, our scheme showes that $\alpha$ should decrease with the increase of the users' caching size. When users' caching size is sufficiently large, only one user should be allowed to send information, indicating that the cooperation network can be just a simple shared link  connecting all users.  	  {In the decentralized random caching case,  $\alpha$ is dynamically changing during the delivery phase, according to the varying sizes of subfiles created in the placement phase.} In other words,  alway letting more users parallelly send information can  cause high transmission delay. 
		\end{itemize}
		
	Note that in the scheme proposed in \cite{D2D}, due to the fixed locations of users  in the network topology, each user connects with a fixed set of users, and  users' cache sizes must be large enough to store all files in the library. While in our schemes, the users' group partition is dynamically changing, and each user can communicate with any set of  users through network routing. These  differences breaks the restriction of users' cache size $M$ and leads to  divergent    data delivery design.  Besides, our model  has the server  share  communication loads with the  users, resulting in an allocation problem on the communication loads between the server and users. Finally, our schemes   achieve   a trade-off between the cooperation gain, parallel gain and  multicast gain, while the schemes  in \cite{D2D,Centralized,Decentralized,D2D} only achieve  the multicast gain. 
		 		


	The remainder of this paper is as follows. Section \ref{Sec_Model} presents the system model, and defines the main problem studied in this paper. We summarize the obtained main results in Section \ref{Sec_Results}. Followed is the detailed description of the centralized coded caching scheme with user cooperation in Section  \ref{Sec_Schemes}. Section \ref{Sec_Schemes_Decen}
		extends the techniques
	we developed for the centralized caching problem to  the setting of decentralized random caching. Section \ref{Sec_Conclusion} concludes this paper.
	
	\begin{figure}
		\centering
		\includegraphics[width=0.59\textwidth,scale=0.2]{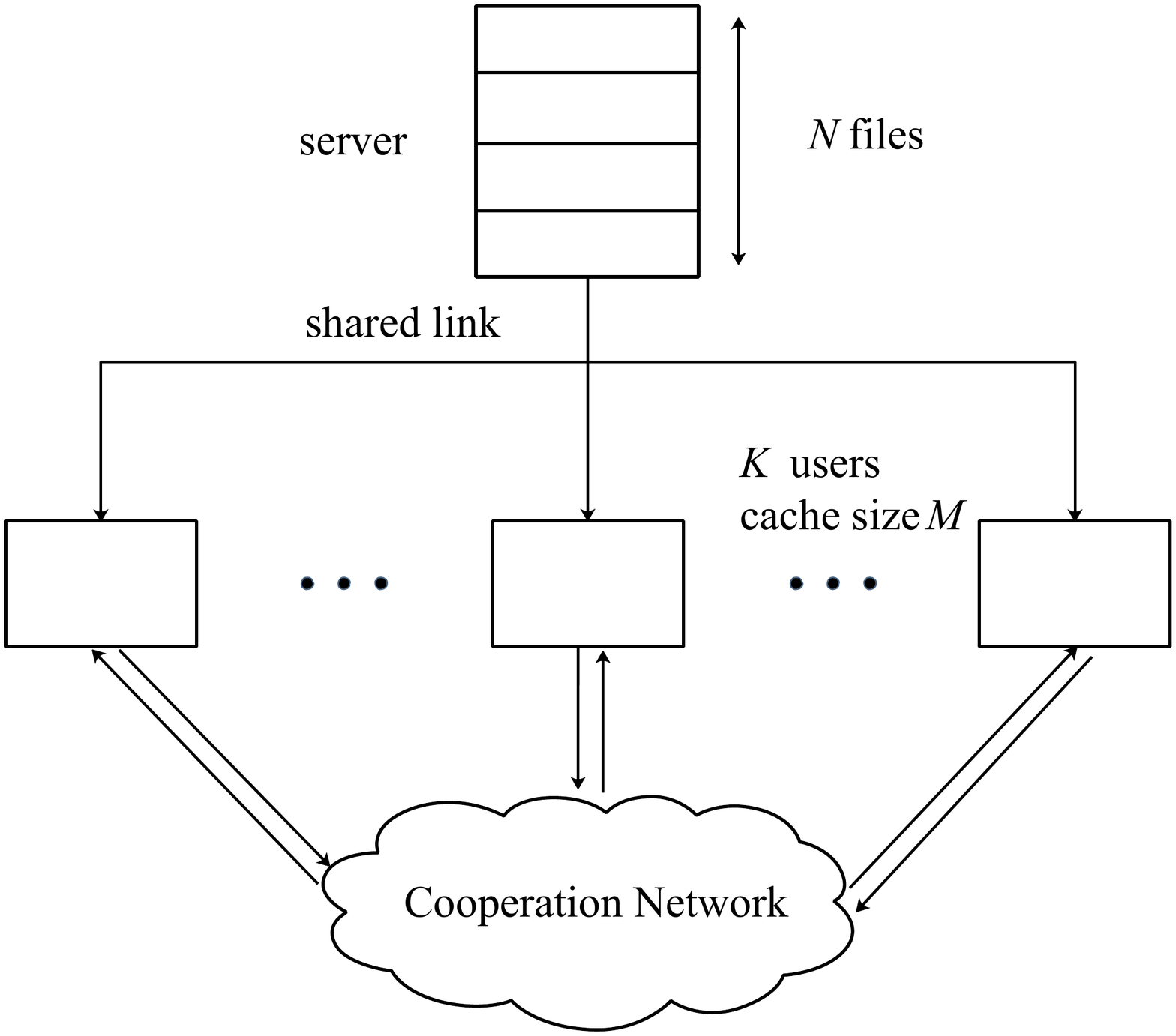}
		\caption{Caching system considered in this paper. A  server connects with $K$ cache-enabled users and the users can cooperate through a flexible network. 
		}
		\label{fig_model}
	\end{figure}

	\section{System Model and Problem Definition}\label{Sec_Model}
	
	Consider a cache-aided network consisting of a single server and $K$ users  as depicted in Fig. \ref{fig_model}.   The server  has a library of $N$ independent files $W_1,\ldots,W_N$. Each file $W_n$, $n=1,\ldots,N$, is uniformly distributed over \[[2^F]\triangleq\{1,2,\ldots,2^F\},\] for some positive integer $F$. The server connects with $K$ users through a noiseless shared link  but  rate-limited to $C_1$ bits per second.  Each user $k\in[K]$  is equipped with a cache memory of size $MF$ bits, where $M\in [0,N]$, and can communicate with each other through a  cooperation network.  
	
	We mainly focus on two types of cooperation networks: a shared   link as in \cite{Centralized, Decentralized} and a flexible routing network   introduced in \cite{multi-server}. In the case of the shared link, all users connect with each other through a shared error-free link but   rate-limited to $C_2$ bits per seconds. In the flexible routing network,    $K$ users can arbitrarily form  into multiple groups,  in each of which    one user    sends  data packets via network routing to the remaining users  in the corresponding  group. The transmission in each group is  error-free in $C_2$ bits per second and interference-free from other groups' transmission signals. 
	  To include   various   types of cooperation networks, we introduce an   integer $ \alpha_{\max}\in\mathbb{Z}$, which denotes the maximum number of groups allowed to send data parallelly in the cooperation network.  For example, when  $\alpha_\text{max}=1$, the cooperation network degenerates into a    shared link, and when   $\alpha_\text{max}=\lfloor \frac{K}{2}\rfloor$, it turns to the  flexible network. In this paper, we consider the general case  $ \alpha_{\max}\in[\lfloor\frac{K}{2}\rfloor]$.


	The system  works in two phases: a placement phase and a delivery phase. In the placement phase, all users will access the entire library $W_1,\ldots,W_N$ and fill the content to their caches. More specifically,  each user $k$, for $k\in[K]$, maps  $W_1,\ldots,W_N$ to its cache contents: 
	\begin{IEEEeqnarray}{rCl}
		Z_k \triangleq \phi_k(W_1,\ldots,W_N),
	\end{IEEEeqnarray}
	for some caching function
	\begin{IEEEeqnarray}{rCl}\label{eq:caching}
		\phi_{k}: [2^F]^N\rightarrow [\lfloor2^{MF}\rfloor].
	\end{IEEEeqnarray}

	In the delivery phase,  each user  requests one of the $N$ files from the library. We denote the demand of  user $k$ {as}
	$d_k\in[N]$,
	and its demanding file  {as}  $W_{d_k}$. Let 
	$\mathbf{d}\triangleq(d_1,\ldots,d_K)$ denote the users' request vector. 	In this paper, we investigate     the worst request case where each of the users makes unique request. 
				 
		 After  users' requests $\bf{d}$ are informed to  the server and all users,  the server produces symbol
	\begin{IEEEeqnarray}{rCl}
	X\triangleq f_{\bf{d}}(W_1,\ldots,W_N),
	\end{IEEEeqnarray}
	and user $k\in\{1,\ldots,K\}$ produces symbol\footnote{Each user $k$ can produce $X_k$ as a function of $Z_k$ and the received signals sent by the server, but because all users can access to the server's signal due to the fact that the server broadcasts its signals to the network, it's equivalent to generate $X_k$  as a function $Z_k$.}
	\begin{IEEEeqnarray}{rCl}
		X_k\triangleq f_{k,{\bf{d}}}(Z_k),
	\end{IEEEeqnarray}
	for some encoding functions
	\begin{subequations}\label{eq:encoding}
		\begin{IEEEeqnarray}{rCl}
			&&f_{\bf{d}}:[2^F]^N\rightarrow [\lfloor 2^{R_1F} \rfloor],\\
			&& f_{k,\bf{d}}: [\lfloor 2^{MF} \rfloor]\rightarrow  [\lfloor 2^{R_2F} \rfloor],\quad \label{eq:userencoding}
		\end{IEEEeqnarray}
	\end{subequations}
	where $R_1$ and $R_2$ denote the \emph{transmission rate} sent by the server and each user, respectively.    Here we focus on the symmetric case where all users have the same transmission rate.  Due to the constraint of $\alpha_\text{max}$, at most $\alpha_\text{max}$ users can send signals parallelly in each channel use. The set of $\alpha_\text{max}$ users  who send signals in parallel could be adaptively changed in the delivery design.
	

	At the end of the delivery phase, based on     the signals sent from the server and other users,  user $k$   decodes its desired message as
	\[\hat{W}_{d_k}=\psi_{k,\bf{d}}(X,Y_{k},Z_k),\]
	where $Y_{k}$   denotes  user $k$'s received signals    sent from the server and other users, and  $\psi_{k,\bf{d}}$ is a decoding function. 

	We define the worst-case probability of error as
	\begin{IEEEeqnarray}{rCl}
		P_e \triangleq \max_{{\bf{d}}\in \mathcal{F}^n} \max_{k\in[K]} \text{Pr} \left(\hat{W}_{d_k} \neq {W}_{d_k}\right).
	\end{IEEEeqnarray}

{	
	A caching scheme $(M_1,R_1,R_2)$ consists of caching functions $\{\phi_k\}$, encoding functions $\{  f_{\bf{d}},f_{k,\bf{d}}\}$ and decoding functions $\{\psi_{k,\bf{d}}\}$. We say that the rate region $(M,R_1,R_2)$  is \emph{achievable} if for every $\epsilon>0$ and every large enough file size $F$, there exists a caching scheme such that $P_e$    is less than $\epsilon$.


 
Since the server and the users   send   signals in parallel, the total transmission delay, denoted by $T$, can be defined as 
\begin{IEEEeqnarray}{rCl}\label{qeDelay}
T \triangleq \max \{\frac{R_1F}{C_1},\frac{R_2F}{C_2}\}.
\end{IEEEeqnarray}
The \emph{optimal}  transmission delay is $T^*\triangleq \inf\{ T: T~\text{is achievable}\}$. 
For simplicity, we  assume that   $C_1=C_2=F$, and then from \eqref{qeDelay} we have
\begin{IEEEeqnarray}{rCl}\label{eq:fDuplex}
T= \max\{R_1,R_2\}.
\end{IEEEeqnarray}

{

	 	}}

	Our goal is to design coded caching schemes that minimize the transmission delay.
	Finally,  in this paper we assume
	$K \leq N$ and $M \leq N$. Extending the results to other scenarios is straightforward, as  mentioned in \cite{Centralized}.
	
	\section{Main Results}\label{Sec_Results}
	For the system model described in Section \ref{Sec_Model}, we first 
establish a lower bound on the transmission delay, then present  new upper bounds  and optimality results of  our centralized and decentralized coded caching schemes, respectively. 
	

	\begin{Theorem}[Lower Bound]\label{Thrm_LowerBound}
			For  memory size $0 \leq M \leq N$, the optimal transmission delay is lower bounded by
			\begin{IEEEeqnarray}{rCl}\label{eq:cutset}
				T^*\geq \max&&\left\{\frac{1}{2}\Big(1-\frac{M}{N}\Big),\max\limits_{s\in [K]}\Big(s-\frac{KM}{\lfloor N/s\rfloor}\Big), \right.\nonumber\\&&\qquad\left.\max\limits_{s\in [K]}\Big(s-\frac{sM}{\lfloor N/s\rfloor}\Big)\frac{1}{1+\alpha_\textnormal{max}}\right\}.\label{eq:cutset2}
			\end{IEEEeqnarray}
		\end{Theorem}
		\begin{proof}
			See the proof in Appendix \ref{App_converse}.
		\end{proof}
	
	
	\subsection{Centralized  Coded Caching}
	 In the following Theorem, we present an upper bound on the transmission delay for the centralized  caching setup.
	\begin{Theorem}[Upper Bound of the 	Centralized Scheme]\label{Thrm_UpperBound}
		Let $t\triangleq KM/N\in \mathbb{Z}^+$, and $\alpha\in \mathbb{Z}^+$. For  memory size $M\in\{0,\frac{N}{K},\frac{2N}{K},\ldots,N\}$, the optimal transmission delay $T^*$ is upper bounded by $T^*\leq T_\textnormal{central}$, where 
		\begin{IEEEeqnarray}{rCl}\label{eq:achievable_rate}
			 T_\textnormal{central} \triangleq \!\!\min_{ \alpha \leq  \alpha_{\max}  }\!\! K\Big(1\!-\!\frac{M}{N}\Big)\frac{1}{1\!+\!t\!+\!{\alpha\min\{\lfloor  \frac{K}{\alpha}\rfloor\!-\!1,t\}}}.\quad 
		\end{IEEEeqnarray}
		For general $0 \leq M \leq N$, the lower convex envelope of these points is achievable. 
	\end{Theorem}

	\begin{proof}
		See  scheme in Section \ref{Sec_Schemes}.
	\end{proof}

	The following simple example   shows that the proposed upper bound   can greatly reduce the transmission delay.
\begin{Example}
		Consider a network described in Section \ref{Sec_Model} with $KM/N=K-1$. The coded caching scheme without user cooperation in \cite{Centralized} allows the server to create an XOR message useful for all $K$ users, achieving the transmission delay  $K\big(1-\frac{M}{N}\big)\frac{1}{1+t}= \frac{1}{K}$. The coded caching scheme without server in  \cite{D2D} achieves the transmission delay  $\frac{N}{M}(1-\frac{M}{N})=\frac{1}{K-1}$. Our upper bound in Theorem \ref{Thrm_UpperBound} achieves  $\frac{1}{2K-1}$ by choosing $\alpha =1$, 		which reduces the transmission delay by around 2 times when $K$ is large.
	\end{Example}

	From \eqref{eq:achievable_rate}, we obtain the optimal value of $\alpha$, denoted by $\alpha^*$,   equals to 1 if $t\geq K-1$ and   to $\alpha_\textnormal{max}$ if $ t \leq \lfloor  \frac{K}{\alpha_\textnormal{max}}\rfloor\!-\!1$. When ignoring all integer constraints, we obtain  $\alpha^*=\frac{K}{t-1}$. We rewrite this choice as follows: 
	\begin{equation} \label{eq:optimalAlpha}
	\alpha^*=\!\left\{
	\begin{aligned}
	&1, &t\geq K-1,\\
	&\frac{K}{t+1},  &\lfloor  \frac{K}{\alpha_\textnormal{max}}\rfloor\!-\!1 \!<\!t\!<\! K\!-\!1,\\
	&\alpha_\textnormal{max}, & t \leq \lfloor  \frac{K}{\alpha_\textnormal{max}}\rfloor\!-\!1.
	\end{aligned}
	\right.
	\end{equation}
	
	\begin{Remark}
	Recall that parameter $\alpha$ denotes the number of users that exactly send information parallelly in the delivery phase. It's interesting to see that $\alpha$ should decrease as  the users' caching size $M$ increases for given  $(K, N, \alpha_{\max})$. To simplify the explanation, we assume $\alpha_\textnormal{max}=\lfloor \frac{K}{2}\rfloor$ and  $KM/N\in \mathbb{Z}^+$.
	When $M\leq N(\lfloor  \frac{K}{\alpha_\textnormal{max}}\rfloor-1)/K$, we have $\alpha^*=\alpha_\textnormal{max}$ and thus it's beneficial to let the most   users  parallelly send information.  As  $M$ increases,  $\alpha^*$ decreases and $\alpha^*< \alpha_\textnormal{max}$,  indicating that  letting more  users  parallelly send information could be harmful.  Too see this, consider the case when $N=100$, $K=10$, $\alpha_\textnormal{max}=5$, $M=40$. From \eqref{eq:optimalAlpha}, we  $\alpha^* =2<\alpha_\textnormal{max}$, smaller than $\alpha_\textnormal{max}=5$.  In the extreme case when $M\geq (K-1)N/K$,  only one user should be allowed to send information,   implying that when users' caching size is sufficiently large, the  cooperation network can be just a simple shared link  connecting with all users.  The main reason for this phenomenon is due to a tradeoff between the multicast gain, \emph{cooperation gain} and \emph{parallel gain}, which will be introduced later in this section.
\end{Remark}

	
	
	
	
	Comparing $ T_\textnormal{central}$ with the transmission delay achieved by the  scheme without user cooperation in \textnormal{\cite{Centralized}}, i.e., $K\big(1-\frac{M}{N}\big)\frac{1}{1+t}$,    $ T_\textnormal{central}$ consists of an additional factor 
	\begin{IEEEeqnarray}{rCl}\label{CCGainCen}
	G_\textnormal{central,c}\triangleq \frac{1}{1+{\frac{\alpha}{1+t}\min\{\lfloor  \frac{K}{\alpha}\rfloor\!-\!1,t\}}},
\end{IEEEeqnarray}
	  referred to   \emph{centralized cooperation gain},  as it arises from   user cooperation.  Comparing  $ T_\textnormal{central}$  with  the delay achieved by the scheme for D2D network without server \cite{D2D}, i.e.,  $\frac{N}{M}(1-\frac{M}{N})$,   $ T_\textnormal{central}$ consists of an additional factor 
	\begin{IEEEeqnarray}{rCl}\label{CPGainCen}
	G_\textnormal{central,p} \triangleq\frac{1}{1+\frac{1}{t}+\frac{\alpha}{t}\min\{\lfloor  \frac{K}{\alpha}\rfloor\!-\!1,t\}},
\end{IEEEeqnarray}
	referred to \emph{centralized parallel gain},  as it arises  from    parallel transmission among the server and users. Both gains  depend on $K$, $M/N$ and $\alpha_{\max}$. 
	
	Subsisting the optimal $\alpha^*$  into  \eqref{CCGainCen}, we have
	\begin{equation} \label{eq:cen_cooperation_gain}
	G_\textnormal{central,c}=\!\left\{
	\begin{aligned}
	&\frac{1+t}{K+t}, &t\geq K-1,\\
	&\frac{1+t}{(\lfloor  \frac{K}{\alpha^*}\rfloor\!-\!1)\alpha^*\!+\!t\!+\!1}, &\lfloor  \frac{K}{\alpha_\textnormal{max}}\rfloor\!-\!1 \!<\!t\!<\! K\!-\!1,\\
	&\frac{1+t}{\alpha_\textnormal{max}t+t+1}, & t \leq \lfloor  \frac{K}{\alpha_\textnormal{max}}\rfloor\!-\!1.
	\end{aligned}
	\right.
	\end{equation}
 When fixing $(K, N, \alpha_{\max})$, $G_\textnormal{central,c}$ in general is not a  monotonic function of $M$. 	More specifically,  when $M$ is small such that $t<  \lfloor \frac{K}{\alpha_\textnormal{max}}\rfloor\!-\!1$, the function $G_\textnormal{central,c}$ is monotonically decreasing, indicating  that the improvement caused by user cooperation increases. This is mainly because relatively larger $M$ allows users to share more common data with each other,  providing more opportunities on user cooperation. However, when $M$ gets larger such that $t \geq \lfloor \frac{K}{\alpha_\textnormal{max}}\rfloor\!-\!1$, the local and global caching gains become dominant, and less improvement can be obtained from user cooperation, turning  $G_\textnormal{central,c}$  to a monotonic increasing function of $M$, 
	
	

	Similarly, subsisting the optimal $\alpha^*$  into  \eqref{CPGainCen}, we  obtain
	\begin{equation} \label{eq:cen_parallel_gain}
	G_\textnormal{central,p}=\!\left\{
	\begin{aligned}
	&\frac{t}{K+t}, &t\geq K-1,\\
	&\frac{t}{\alpha^* t+t+1},&\lfloor  \frac{K}{\alpha_\textnormal{max}}\rfloor\!-\!1 \!<\!
	t\!<\! K\!-\!1,\\
	&\frac{t}{\alpha_\textnormal{max}t+t+1}, & t \leq \lfloor  \frac{K}{\alpha_\textnormal{max}}\rfloor\!-\!1.
	\end{aligned}
	\right.
	\end{equation}
	Eq. \eqref{eq:cen_parallel_gain} shows that $G_\textnormal{central,p}$ is monotonically increasing referring to $t$,  mainly due to the fact that as $M$ increases,  more contents can be sent through the user cooperation without the help of the central server, decreasing the improvement from parallel transmission between the server and users. 
	
	The centralized cooperation gain \eqref{CCGainCen} and parallel gain \eqref{CPGainCen} are plotted in  Fig. \ref{fig_gain} when $N=20$, $K=10$ and $\alpha_{\max}=5$. 
	
		
		\begin{Remark}
	 Larger $\alpha$ could  lead to better parallel and cooperation gain (more uses can concurrently multicast signals to other users), but may result in   worse multicast gain (signals are multicasted to less users in each group),  the  choice of $\alpha$ in \eqref{eq:optimalAlpha} is in fact a  tradeoff  between the multicast gain, parallel gain and cooperation gain. 
	\end{Remark}

 		The proposed upper bound in Theorem \ref{Thrm_UpperBound} is order optimal.
		\begin{Theorem}\label{Thrm_Gap}
			For  memory size $0 \leq M \leq N$, 
			\begin{IEEEeqnarray}{rCl}
				\frac{ T_\textnormal{central}  }{T^*}\leq 31.
			\end{IEEEeqnarray}
		\end{Theorem}
		\begin{proof}
			See the proof in Appendix \ref{App_Gap}.
		\end{proof}

	\begin{figure}
		\centering
		\includegraphics[width=0.5\textwidth,scale=0.2]{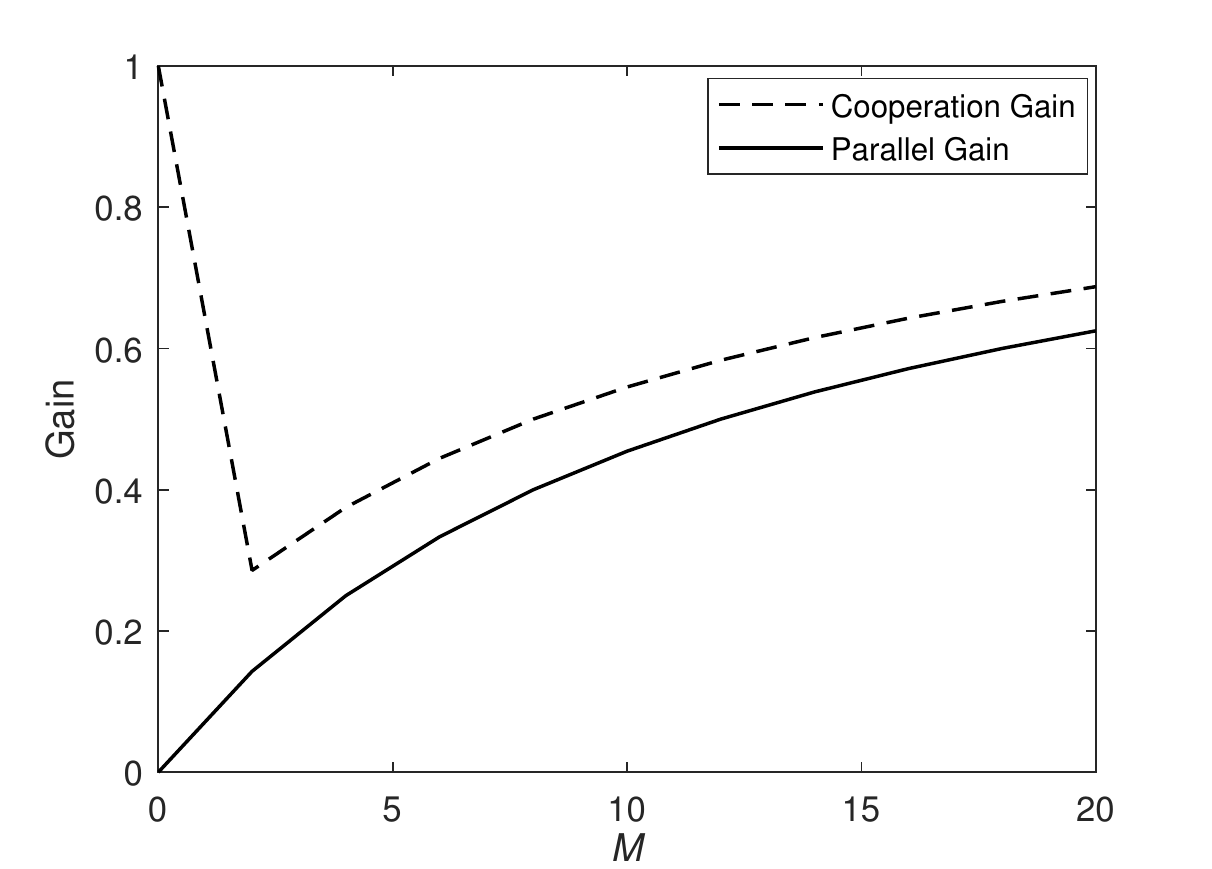}
		\caption{Centralized cooperation gain and parallel gain  when $N=20$, $K=10$ and $\alpha_{\max}=5$.  
		}
		\label{fig_gain}
	\end{figure}

		\begin{figure}
			\centering
			\includegraphics[width=0.5\textwidth,scale=0.2]{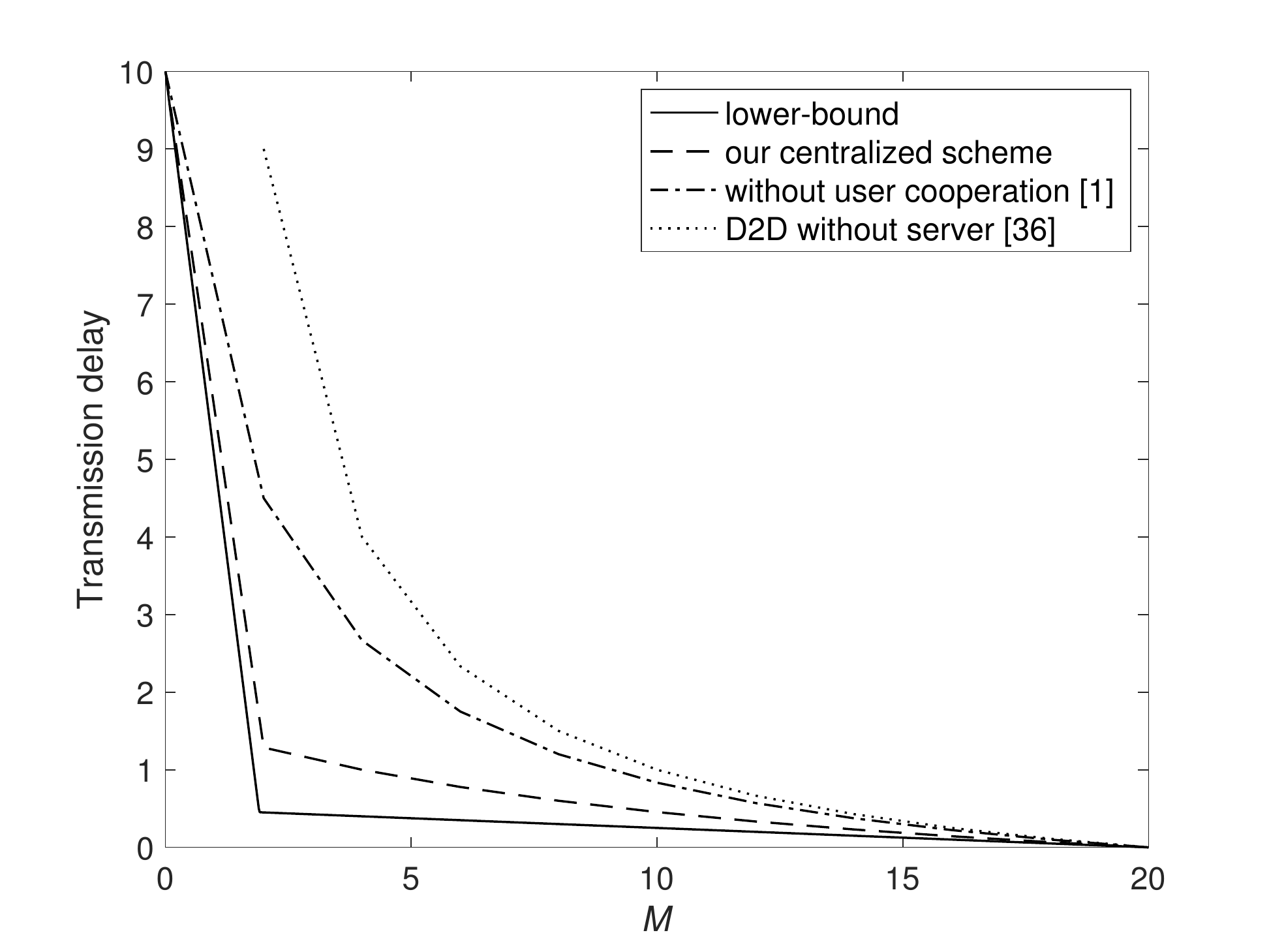}
			\caption{
				Transmission delay when $N=20$, $K=10$ and $\alpha_{\max}=5$.  The upper bounds  are achieved under the centralized  caching scenario.
			}
			\label{fig_result}
		\end{figure}
{The exact gap 	of ${ T_\textnormal{central}  }/{T^*}$ could be much smaller. One could apply the method  proposed in \cite{ImprovedGap2} to obtain a tighter lower bound and shrink the gap. In this paper, we only prove the order optimality of the proposed schemes, and   leave the work of finding a smaller gap into the future.  

Fig. \ref{fig_result} plots the lower bound \eqref{eq:cutset} and upper bounds achieved by various schemes, including the proposed scheme, the scheme in \cite{Centralized} which considers the setting without user cooperation, and the scheme in \cite{D2D} which considers the setting without server. It's obvious that our scheme outperforms the previous schemes and approaches closely to the lower bound.
		}

	\subsection{Decentralized Coded Caching}

We exploit multicast gain from coded caching,  cooperation opportunity among users, and  parallel transmission between the server and users,  and      achieve an upper bound  stated below. 
	\begin{Theorem}[Upper Bound of the Decentralized Scheme]\label{Thrm_UpperBound_Decen}
		Define $p\triangleq M/N$. For memory size $0 \leq M \leq N$,  the optimal transmission delay  $T^*$ is upper bounded by $ T_\textnormal{decentral}:$
		\begin{subequations}\label{eq:achievable_rate_decen_whole}
			\begin{IEEEeqnarray}{rCl}\label{eq:achievable_rate_decen}
				{T}_\textnormal{decentral}\triangleq \max \left\{R_\emptyset, \frac{R_\textnormal{s}R_\textnormal{u}}{R_\textnormal{s}+R_\textnormal{u}-R_{\emptyset}}\right\},
			\end{IEEEeqnarray}
			where 
			\begin{IEEEeqnarray}{rCl}
				R_{\emptyset} & \triangleq& K(1-p)^K,\label{eq:Re}\\
				R_\textnormal{s} & \triangleq& \frac{1-p}{p}\big(1-(1-p)^K\big),\label{eq:Rs}\\ 
				R_\textnormal{u} &\triangleq &  \frac{1}{\alpha_\textnormal{max}}  \sum_{s=2}^{\lceil \frac{K}{\alpha_\textnormal{max}}\rceil-1}\frac{s{K \choose s}}{s-1}p^{s-1}(1-p)^{K-s+1}   \nonumber \\ &&\hspace{6pt}+\sum_{s= \lceil \frac{K}{\alpha_\textnormal{max}}\rceil}^{K}     \frac{K {K-1 \choose s-1}}{f(K,s)} p^{s-1}(1-p)^{K-s+1},\label{eq:Ru}  	
			\end{IEEEeqnarray}
			with
			\begin{IEEEeqnarray}{rCl}\label{eq:f(k,s)}
				f(K,s) \triangleq
				\left\{
				\begin{aligned}
					&\lfloor \frac{K}{s} \rfloor(s-1), & (K~\textnormal{mod}~s)<2,\\
					&K-1-\lfloor{K}/{s}\rfloor, &(K~\textnormal{mod}~s)\geq2.
				\end{aligned}
				\right.
			\end{IEEEeqnarray}
		\end{subequations}
	\end{Theorem}
	\begin{proof}
Here $R_{\emptyset}$ represents the transmission rate of sending content that are not cached by any user,  $R_\textnormal{s}$ and $R_\textnormal{u}$ represent the transmission rate caused by the server when sending data in case of no user cooperation,  and the transmission rate caused by each user when sending messages with the absence of server, respectively. 
 Eq. (\ref{eq:achievable_rate_decen}) balances out the communication loads assigned to the server and  users. 
 See more detailed proof  in Section \ref{Sec_Schemes_Decen}.
 \end{proof}

		\begin{Remark}
The upper bound in Theorem \ref{Thrm_UpperBound_Decen} is achieved by setting the  number of users that exactly  send signals in parallel as follows:
		\begin{equation} \label{eq:optimalAlphaD}
		\alpha_\textnormal{D}=\!\left\{
		\begin{aligned}
		&\alpha_{\max},   \ &\text{case 1}, \\
		&\lfloor \frac{K}{s} \rfloor,  \  &\text{case 2},\\ 
				&\lceil \frac{K}{s} \rceil,   \  &\text{case 3}.\\
		\end{aligned}
		\right.
		\end{equation}
If $\lceil \frac{K}{s} \rceil > \alpha_{\max}$, the  number of users that actually send data in parallel is  smaller than $\alpha_\text{max}$, indicating that alway letting more users parallelly send messages could cause higher transmission delay. This can be easily seen when $s=K-1$, and $\alpha_\text{max}= \lceil \frac{K}{s} \rceil$, as in this scenario  it's sufficient to let a single user $k\in[K]$ broadcast XOR symbol to all $K-1$ users in each transmission slot.  
		\end{Remark}
 
	{
		\begin{Remark}\label{Re_monocity}\label{Coro_Monotonicity}
	From the definitions of $T_\textnormal{decentral}$, $R_\textnormal{s}$, $R_\textnormal{u}$ and $R_{\emptyset}$,		it's easy to  obtain $R_\emptyset\leq T_\textnormal{decentral}\leq R_\textnormal{s}$, and $T_\textnormal{decentral}$ {decreases} as ${\alpha_\textnormal{max}}$ increases. 
	\begin{IEEEeqnarray}{rCl}\label{rateRemark}
	{T}_\textnormal{decentral}=
	 \left\{
	\begin{aligned}
	&\frac{R_\textnormal{s}R_\textnormal{u}}{R_\textnormal{s}+R_\textnormal{u}-R_{\emptyset}}, & R_\textnormal{u}\geq R_\emptyset,\\
	&R_\emptyset , &R_\textnormal{u}< R_\emptyset.
	\end{aligned}
	\right.
    \end{IEEEeqnarray}
    and $T_\textnormal{decentral}$ increases as $R_\textnormal{u}$ increases   if $R_\textnormal{u}\geq R_\emptyset$.
		\end{Remark}


		Due to the complex term $R_\text{u}$, the upper bound $T_\textnormal{decentral}$ in Theorem \ref{Thrm_UpperBound_Decen} is hard to evaluate.   Since  $T_\textnormal{decentral}$ is increasing as $R_\textnormal{u}$ increases as Remark \ref{Re_monocity} indicates,  substituting the following upper bound of $R_\textnormal{u}$ into \eqref{eq:achievable_rate_decen_whole}    provides an efficient way to evaluate $T_\textnormal{decentral}$.
		
	} 
	\begin{Corollary}\label{Coro_UpperBound_Decen}
		For memory size $0 \leq p \leq 1$, the upper bound of $R_\textnormal{u}$ is given below:
		\begin{itemize}
			\item $\alpha_{\textnormal{max}} = 1$ (shared link):
			\begin{IEEEeqnarray}{rCl}\label{eq: Shrd_UpperBound}
				R_\textnormal{u} \leq && \bar{R}_{\textnormal{u-s}} \triangleq \frac{1-p}{p}\left[1-\frac{5}{2}Kp\big(1-p\big)^{K-1} \nonumber\right. \\
				&& \hspace{8pt}  \left.-4\big(1-p\big)^K+\frac{3(1-(1-p)^{K+1})}{(K+1)p}\right];
			\end{IEEEeqnarray}
			\item  $\alpha_{\textnormal{max}} = \lfloor \frac{K}{2} \rfloor$:
			\begin{IEEEeqnarray}{rCl}\label{eq: Flex_UpperBound2}
				\hspace{-0pt}	R_\textnormal{u} &\leq& \bar{R}_{\textnormal{u-f}} \triangleq \frac{K(1-p)}{(K-1)}  \bigg[1-\big(1-p\big)^{K-1}  \nonumber \\
				&& \quad\!- \frac{2/p}{K\!-\!2}\big(1\!-\!(1\!-\!p)^K \!-\!Kp(1\!-\!p)^{K\!-\!1}\big)\bigg];\quad 
			\end{IEEEeqnarray}
			
			\item $1< \alpha_{\textnormal{max}} < \lfloor \frac{K}{2} \rfloor$:
			\begin{IEEEeqnarray}{rCl}\label{Semi-Flex-UpperBound}
				R_\textnormal{u} \leq \bar{R}_\textnormal{u} \triangleq \bar{R}_\textnormal{u-s}/\alpha_\textnormal{max} + \bar{R}_\textnormal{u-f}.
			\end{IEEEeqnarray}
		\end{itemize}
	\end{Corollary}
	\begin{proof}
		See the  proof in Appendix \ref{App_Bound_Decen}.
	\end{proof}
	{
		
		Recall that the transmission delay achieved by the decentralized scheme without user cooperation in \textnormal{\cite{Decentralized}} is equal to $R_\textnormal{s}$ given in \eqref{eq:Rs}. We define the ratio between $T_\textnormal{decentral}$ and $R_\textnormal{s}$ as \emph{decentralized cooperation gain}:
		\begin{IEEEeqnarray}{rCl}\label{denceCoGain}
		G_\textnormal{decentral,c} \triangleq \max\{\frac{R_\emptyset}{R_\textnormal{s}},\frac{R_\textnormal{u}}{R_\textnormal{s}+R_\textnormal{u}-R_\emptyset}\},
\end{IEEEeqnarray}
		and $G_\textnormal{decentral,c}$ is in $[0,1]$ by Remark \ref{Re_monocity}. Similar to the centralized scenario, this gain arises from the coordination between users in the cooperation network.  Moreover, we also compare $T_\textnormal{decentral}$  with the transmission delay $(1-p)/p$, achieved by the decentralized scheme for D2D network in \cite{D2D}, and define the ratio between $R_\textnormal{s}$ and  $(1-p)/p$ as \emph{decentralized parallel gain}:
		\begin{IEEEeqnarray}{rCl}\label{dencePaGain}
		G_\textnormal{decentral,p} \triangleq G_\textnormal{decentral,c}\cdot\Big(1-(1-p)^K \Big),
\end{IEEEeqnarray}
	where $G_\textnormal{decentral,p} \in [0,1]$ arises from the parallel transmission between the server and the users. 
	 
	 We plot the decentralized cooperation gain and parallel gain for different types of cooperation networks 
	 in  Fig. \ref{fig:decen_gain} when   $N=20$ and  $K=10$. It can be seen that  $G_\textnormal{decentral,c}$ and  $G_\textnormal{decentral,p}$ in general are not monotonic functions of $M$. 
		Here $G_\textnormal{decentral,c}$ performs similarly to $G_\textnormal{central,c}$. When $M$ is small, the function $G_\textnormal{decentral,c}$ is monotonically decreasing from value 1 until it reaches the minimum. For larger $M$, the function $G_\textnormal{decentral,c}$ turns to  monotonically increase.  The reason for this phenomenon is that in the decentralized scenario, as $M$ increases, the proportion of subfiles that are not cached by any user and must be sent by the server is decreasing. Thus, there are more subfiles that can be sent parallelly by the user cooperation as $M$ increases. In the meanwhile, the decentralized  scheme in \cite{Decentralized} offers an additional multicasting gain. Therefore, we need to tradeoff between these two gains in order to reduce the   transmission delay.

		The function $G_\textnormal{decentral,p}$ behaves differently as it monotonically increases when $M$ is small. After reaching the maximal value, the function $G_\textnormal{decentral,p}$ decreases monotonically until it meets the local minimum\footnote{The abnormal bend in parallel gain when $\alpha_{\max}=\lfloor \frac{K}{2} \rfloor$ come from a balance effect between the $G_\textnormal{decentral,c}$ and $1-(1-p)^K$ in \eqref{dencePaGain}.}, then $G_\textnormal{decentral,p}$ turns into a monotonic increasing function for large $M$. Similar to the centralized case, as $M$ increases, the impact of parallel transmission among the server and users becomes smaller since more data can be transmitted by the users.
				
		


		\begin{Theorem}\label{Thrm_Gap_Decen}
			Define $p\triangleq M/N$  and $p_\textnormal{th} \triangleq 1-\big(\frac{1}{K+1}\big)^\frac{1}{K-1}$, which tends to 0 as $K$ tends to infinity. For memory size $0 \leq M \leq N$,
			\begin{itemize}
				\item if  $\alpha_{\textnormal{max}} = 1$ (shared link),
				\[\frac{{T}_\textnormal{decentral}}{T^*} \leq 24.\]
				
				\item if $\alpha_{\textnormal{max}} = \lfloor \frac{K}{2} \rfloor$,
				\begin{IEEEeqnarray*}{rCl}\label{eq:Gap_Flex}
					\frac{{T}_\textnormal{decentral}}{T^*} \leq\left\{
					\begin{aligned}
						&\max\left\{6, 2K\Big(\frac{2K}{2K+1}\Big)^{K-1}\right\}, &p < p_\textnormal{th},\quad \\
						&6, & p\geq p_\textnormal{th}.\quad
					\end{aligned}
					\right.
				\end{IEEEeqnarray*}
				\item if $1 < \alpha_{\textnormal{max}} < \lfloor \frac{K}{2} \rfloor$,
				\begin{IEEEeqnarray*}{rCl}\label{eq:Gap_Flex}
					\frac{{T}_\textnormal{decentral}}{T^*}\leq\left\{
					\begin{aligned}
						&\max\left\{ \min \left\{12\Big(1+\alpha_\textnormal{max}\Big),\right.\right.\\ &\hspace{32pt}
						\left.2K\Big(\frac{2K}{2K+1}\Big)^{K-1}\Big\}\right\}, & p < p_\textnormal{th},\\
						&77, & p \geq p_\textnormal{th}.
					\end{aligned}
					\right.
				\end{IEEEeqnarray*}
			\end{itemize}
			
		\end{Theorem}
		\begin{proof}
			See the proof in Appendix  \ref{App_Gap_Decen}.
		\end{proof}
		
			\begin{figure}
			\centering
			\includegraphics[width=0.5 \textwidth,scale=0.2]{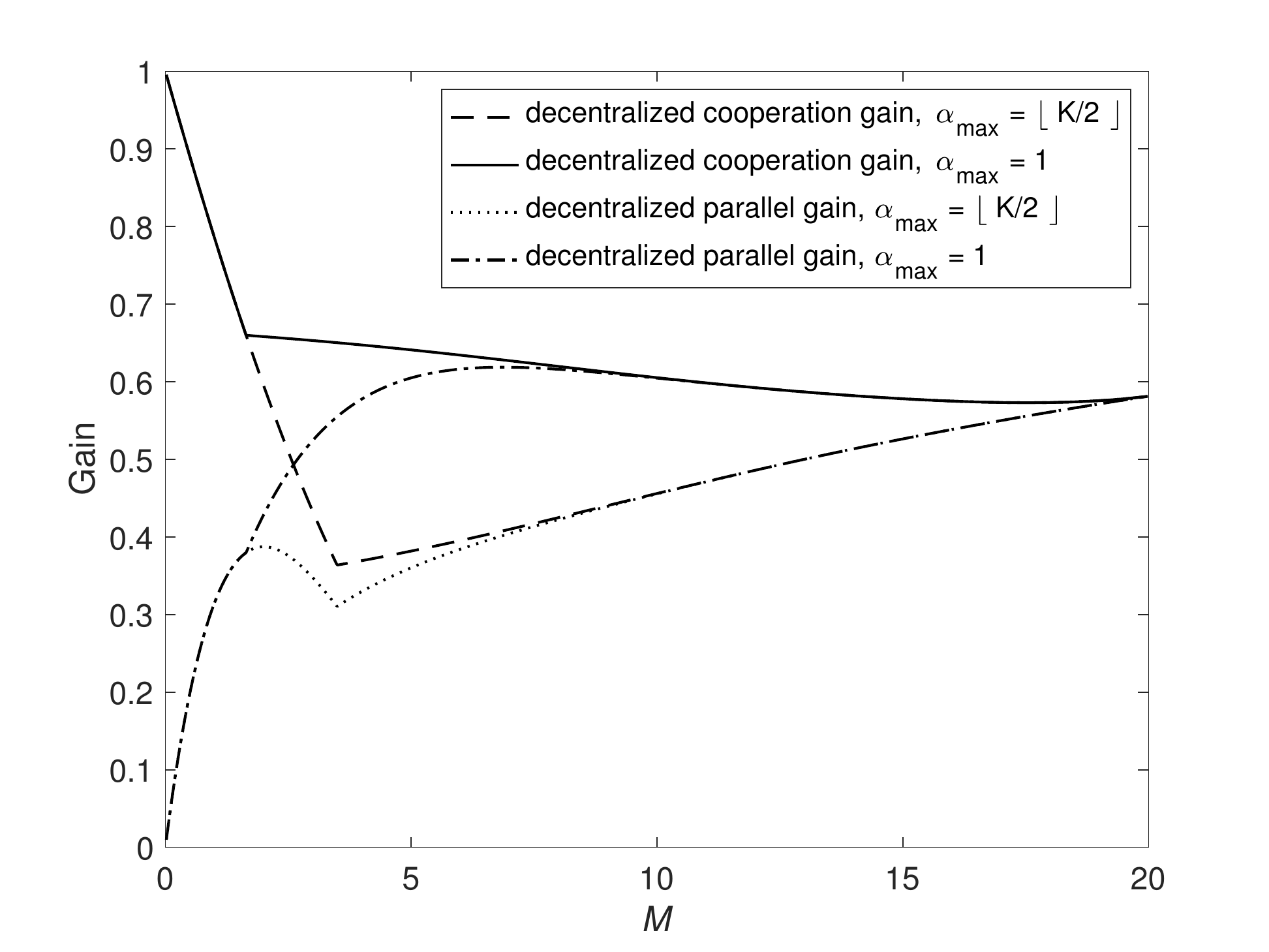}
			\caption{Decentralized cooperation gain and parallel gain when $N=20$ and $K=10$.}
			\label{fig:decen_gain}
		\end{figure}
		\begin{figure}
			\centering
			\includegraphics[width=0.5\textwidth,scale=0.5]{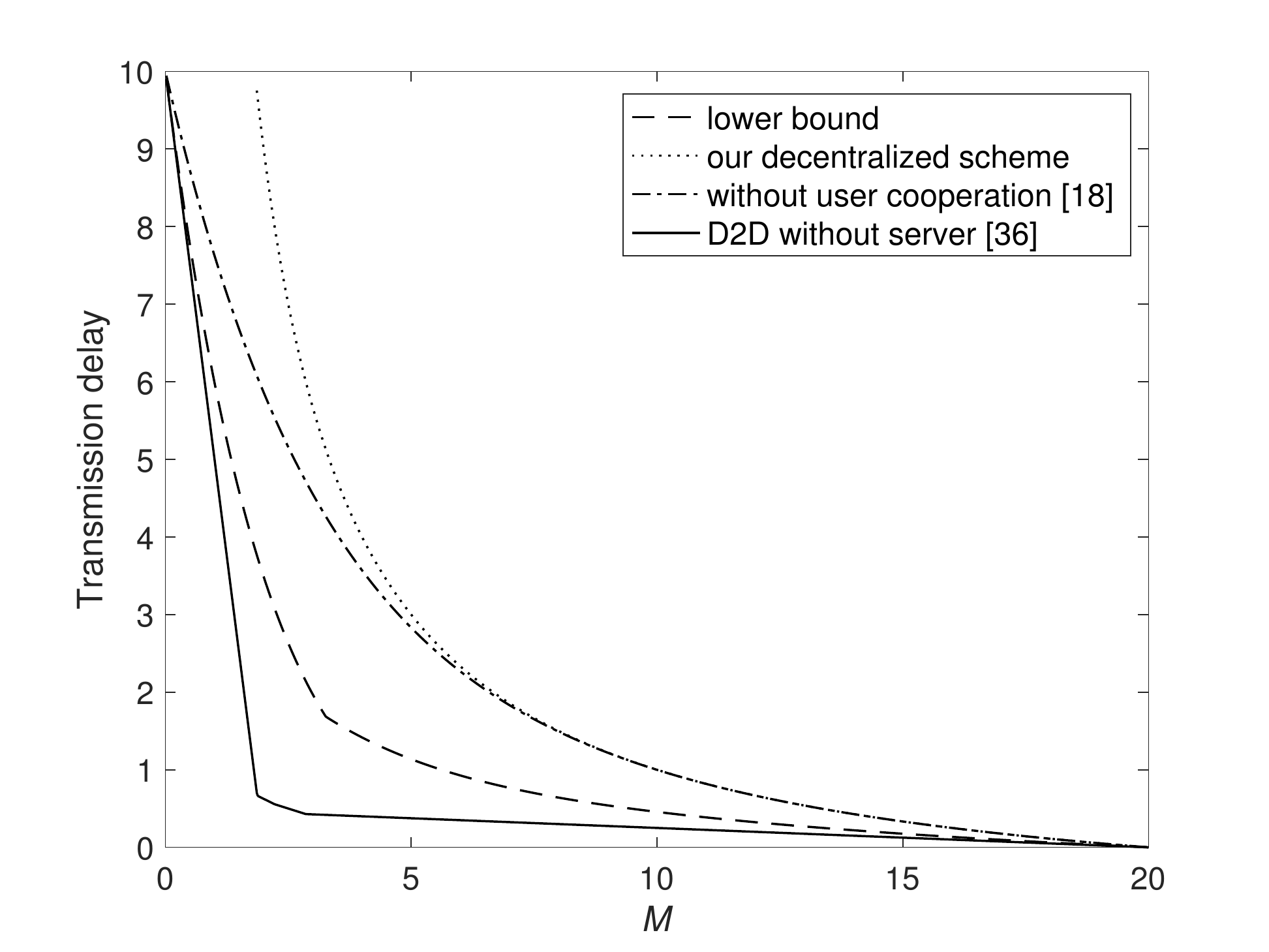}
			\caption{Transmission delay when $N=20$, $K=10$ and $\alpha_{\max}=3$. The upper bounds  are achieved under the decentralized random caching scenario.
			}.
			\label{fig:decen_rate}
		\end{figure}

	{	Fig. \ref{fig:decen_rate}   plots the lower bound  in \eqref{eq:cutset} and  	upper bounds achieved by various decentralized coded caching schemes, including our scheme, the scheme in \cite{Decentralized} which considers the case without user cooperation, and the scheme in \cite{D2D} which considers the case without server.

	}						
		
		\section{Coding Scheme Under Centralized  Data Placement} \label{Sec_Schemes}
		In this section, 
				we describe a novel centralized coded caching scheme for arbitrary $K$, $N$ and $M$ such that $t= {KM}/{N}$ is a positive integer.  When $t$ is not an integer,  	we can use a resource sharing scheme as in  \cite{Centralized}. 
		
		Unlike the setting in   \cite{D2D} where all   users are fixed in a mesh network, leading to an unchanging group partition during the delivery phase,   our schemes allows   users' group partition to  dynamically change, and each user to communicate with any set of  users through network routing. These  differences breaks the limitation on  users' cache size $M$ and requires novel   design on   data delivery.  Besides, our model  involves tradeoff between the  communication loads at the server and users, which is not  existing in \cite{D2D}.  
		
		 
		   We first use an illustrative example to show how we group users, split files and delivery data, and then present our generalized centralized coding caching scheme. 
		\subsection{An Illustrative Example}
			Consider a network consisting of $K=6$ users with cache size $M=4$, and a library of $N = 6$ files. Thus $t=KM/N=4$. Let $\alpha=2$, that is to say we separate the 6 users into 2 groups of equal size. The choice of the groups
			is not unique. 
			We choose an integer $L_1=2$ which satisfies $\frac{K\binom{K-1}{t}L_1}{\min\{\alpha(\lfloor K/\alpha \rfloor-1),t\}}=15$ is an integer. \footnote{According to \eqref{eq:optimalAlpha} and \eqref{eq_Acondition}, one optimal choice could be ($\alpha=1$, $L_1=4$, $\lambda=5/9$), here we choose  ($\alpha=2$, $L_1=2$, $\lambda=1/3$) for simplicity, and also in order to demonstrate that even with a suboptimal choice, our scheme still outperforms that in \cite{Centralized} and \cite{D2D}. } 
			Split each file $W_n$, for $n=1,\ldots,N$, into $3\binom{6}{4}=45$ subfiles:
			\[W_n = (W^l_{n,\mathcal{T}}: l\in[3],\mathcal{T}\subset[6], |\mathcal{T}|=4).\]
			
			
			We list all the requested subfiles uncached by the users as follows: for $l=1,2,3$,
			\begin{equation*}
			\begin{aligned}
			&W_{d_1,\{2345\}}^{l},W_{d_1,\{2346\}}^{l},W_{d_1,\{2356\}}^{l},W_{d_1,\{2456\}}^{l},W_{d_1,\{3456\}}^{l};\\ &W_{d_2,\{1345\}}^{l},W_{d_2,\{1346\}}^{l},W_{d_2,\{1356\}}^{l},W_{d_2,\{1456\}}^{l},W_{d_2,\{3456\}}^{l};\\ &W_{d_3,\{1245\}}^{l},W_{d_3,\{1246\}}^{l},W_{d_3,\{1256\}}^{l},W_{d_3,\{1456\}}^{l},W_{d_3,\{2456\}}^{l};\\ &W_{d_4,\{1235\}}^{l},W_{d_4,\{1236\}}^{l},W_{d_4,\{1256\}}^{l},W_{d_4,\{1356\}}^{l},W_{d_4,\{2356\}}^{l};\\ &W_{d_5,\{1234\}}^{l},W_{d_5,\{1236\}}^{l},W_{d_5,\{1246\}}^{l},W_{d_5,\{1346\}}^{l},W_{d_5,\{2346\}}^{l};\\ &W_{d_6,\{1234\}}^{l},W_{d_6,\{1235\}}^{l},W_{d_6,\{1245\}}^{l},W_{d_6,\{1345\}}^{l},W_{d_6,\{2345\}}^{l}.
			\end{aligned}
			\end{equation*}
			
			The users can finish the transmission in different partitions.  Table \ref{table_time} shows one kind of the partition for example and explains how the users send the requested subfiles with superscript $l=1,2$. 
			\begin{table}  
				\centering
				\caption{Subfiles sent by users in different partition,  $l=1,2$}
				\label{table_time}
				\begin{tabular}{cc}  	
					\toprule
					$\{1,2,3\}$&$\{4,5,6\}$	\\
					user 2:  $W_{d_1,\{2345\}}^{1}\!\oplus\! W_{d_3,\{1245\}}^{1}$&user 5: $W_{d_4,\{2356\}}^{1}\!\oplus\! W_{d_6,\{2345\}}^{1}$\\
					user 2:  $W_{d_1,\{2346\}}^{1}\!\oplus\! W_{d_3,\{1246\}}^{1}$&user 5: $W_{d_4,\{1256\}}^{1}\!\oplus\! W_{d_6,\{1245\}}^{1}$\\
					user 1:  $W_{d_2,\{1346\}}^{1}\!\oplus\! W_{d_3,\{1256\}}^{1}$&user 4: $W_{d_5,\{2346\}}^{1}\!\oplus\! W_{d_6,\{1345\}}^{1}$\\
					user 3:  $W_{d_1,\{2356\}}^{1}\!\oplus\! W_{d_2,\{1356\}}^{1}$&user 6: $W_{d_4,\{1356\}}^{1}\!\oplus\! W^1_{d_5,\{1346\}}$\\
					\midrule
					$\{1,2,4\}$&$\{3,5,6\}$\\
					user 2: $W_{d_1,\{2456\}}^{l}\!\oplus\! W_{d_4,\{1235\}}^{l}$&user 5: $W_{d_3,\{1456\}}^{l}\!\oplus\! W_{d_6,\{1235\}}^{l}$\\
					\midrule
					$\{1,4,6\}$&$\{2,3,5\}$\\
					user 6:  $W_{d_1,\{3456\}}^{l}\!\oplus\! W_{d_4,\{1236\}}^{l}$&user 3: $W_{d_2,\{3456\}}^{l}\!\oplus\! W_{d_5,\{1234\}}^{l}$\\
					\midrule	
					$\{1,2,5\}$&$\{3,4,6\}$\\
					user 1: $W_{d_2,\{1456\}}^{l}\!\oplus\! W_{d_5,\{1236\}}^{l}$&user 4:  $W_{d_3,\{2456\}}^{l}\!\oplus\! W_{d_6,\{1234\}}^{l}$\\
					\midrule
					$\{1,2,3\}$&$\{4,5,6\}$\\
					user 3: $W_{d_1,\{2345\}}^{2}\!\oplus\! W_{d_2,\{1345\}}^{2}$&user 4: $W_{d_5,\{2346\}}^{2}\!\oplus\! W_{d_6,\{2345\}}^{2}$\\
					user 3: $W_{d_1,\{2346\}}^{2}\!\oplus\! W_{d_2,\{1346\}}^{2}$&user 4: $W_{d_5,\{1246\}}^{2}\!\oplus\! W_{d_6,\{1245\}}^{2}$\\		  
					user 2: $W_{d_1,\{2356\}}^{2}\!\oplus\! W_{d_3,\{1245\}}^{2}$&user 5: $W_{d_4,\{1356\}}^{2}\!\oplus\! W_{d_6,\{1345\}}^{2}$\\
					user 1: $W_{d_3,\{1246\}}^{2}\!\oplus\! W_{d_2,\{1356\}}^{2}$&user 6: $W_{d_4,\{1256\}}^{2}\!\oplus\! W_{d_5,\{1346\}}^{2}$\\
					user 1: $W_{d_3,\{1256\}}^{2}\!\oplus\! W_{d_2,\{1345\}}^{1}$&user 6: $W_{d_5,\{1246\}}^{1}\!\oplus\! W_{d_4,\{2356\}}^{2}$\\
					\bottomrule  
				\end{tabular}
			\end{table}
			In Table \ref{table_time}, all the users send an XOR symbol of subfiles with superscript  $l=1$  at the beginning. Note that the subfiles  $W_{d_2,\{1345\}}^1$ and $W_{d_5,\{1246\}}^1$ are left since $\frac{K\binom{K-1}{t}}{\alpha(\lfloor K/\alpha \rfloor-1)}$ is not an integer. Similarly, for subfiles with  $l=2$,   $W_{d_3,1256}^2$ and $W_{d_4,2356}^2$ are not sent to user 3 and 4. In the last transmission,  user 1 delivers the XOR message $W_{d_3,\{1256\}}^{2}\oplus W_{d_2,\{1345\}}^{1}$  to user 2 and 3, and user 6 multicasts $W_{d_5,\{1246\}}^{1}\oplus W_{d_4,\{2356\}}^{2}$ to user 5 and 6. The  transmission rate at the users is $R_2=\frac{1}{3}.$
			
			For the remaining subfiles with $l=3$, the server delivers them   in the same way as in \cite{Centralized}. Specifically, it sends symbols $\oplus_{k\in \mathcal{S}} W_{d_k,\mathcal{S}\backslash\{k\}}^{3}$,
			for all $\mathcal{S}\subseteq[K]:|\mathcal{S}|=5$. Thus the rate sent by the server is $R_1=\frac{2}{15}$,
			and the transmission delay $ T_\textnormal{central}=\max\{R_1,R_2\}=\frac{1}{3}$,
			which is less than the delay achieved by the centralized coded caching scheme without user cooperation  $K\big(1-\frac{M}{N}\big)\frac{1}{1+t}=\frac{2}{5}$, and by the centralized coded caching scheme without server   $\frac{N}{M}\big(1-\frac{M}{N}\big)=\frac{1}{2}$.


		
		\subsection{The Gereralized Centralized Coding Caching Scheme}
		In the placement phase, each file is first split into $\binom{K}{t}$ subfiles of equal size, and then each subfile is split into two non-overlaping mini-files whose size could be unequal. More specifically, we split $W_n$ into subfiles as follows:
		\begin{IEEEeqnarray}{rCl}
			W_n = \left(W_{n,\mathcal{T}}:\mathcal{T}\subset [K], |\mathcal{T}|=t\right).
		\end{IEEEeqnarray}
	User $k$ caches all the subfiles when $k \in \mathcal{T}$ for all $n=1,...,N$, occupying  cache size of $MF$ bits.	Then split each subfile $W_{n,\mathcal{T}}$ into two  mini-files 
		\[W_{n,\mathcal{T}} = \Big(W_{n,\mathcal{T}}^{\text{s}}, W_{n,\mathcal{T}}^{\text{u}}\Big )\]
		where  the mini-files $W_{n,\mathcal{T}}^\textnormal{s}$  and $W_{n,\mathcal{T}}^\textnormal{u}$ will be sent by the server and users, respectively, and 
			\begin{equation}
		\begin{aligned}
		|W_{n,\mathcal{T}}^\textnormal{s}| = \lambda \cdot |W_{n,\mathcal{T}}|=\lambda\cdot\frac{F}{\binom{K}{t}},\\
		|W_{n,\mathcal{T}}^\textnormal{u}| = (1-\lambda) \cdot |W_{n,\mathcal{T}}|=(1-\lambda)\cdot\frac{F}{\binom{K}{t}},
		\end{aligned}
		\end{equation}
		with 
		\begin{IEEEeqnarray}{rCl}\label{eq_Acondition}
				 {\lambda}=\frac{1+t}{{\alpha\min\{\lfloor  \frac{K}{\alpha}\rfloor-1,t\}}+1+t}.
			\end{IEEEeqnarray}
			For each  mini-file $W_{n,\mathcal{T}}^\textnormal{u}$, split it into $L_1$   pico-files  of equal size $(1-\lambda)\cdot\frac{F}{L_1\binom{K}{t}}$, i.e., \[W_{n,\mathcal{T}}^{\textnormal{u}}=\left(W_{n,\mathcal{T}}^{\textnormal{u},1},\ldots,W_{n,\mathcal{T}}^{\text{u},L_1}\right),\]where $L_1$ satisfies	
			\begin{IEEEeqnarray}{rCl}\label{eq_Lcondition}
				&&\frac{K\cdot\binom{K-1}{t}\cdot L_1}{\alpha\min\{\lfloor  \frac{K}{\alpha}\rfloor-1,t\}}\in \mathbb{Z}^+.
			\end{IEEEeqnarray}
		As we will see later, condition  \eqref{eq_Acondition} ensures that   communication loads can be optimally allocated  at the server and the users, and \eqref{eq_Lcondition} ensures that the number of subfiles is large enough to   maximum multicast gain when user sending data. 

		

		In the delivery phase, each user $k$ requests file $W_{d_k}$. The requests vector $\mathbf{d}$ is informed by  the server and all the users. Note that different parts of file $W_{d_k}$ have been stored in the users' caches, and  thus the uncached parts of $W_{d_k}$ can   be sent by the server and users.  
		Subfiles
		\[\left(W_{d_k,\mathcal{T}}^{\text{u},1},\ldots,W_{d_k,\mathcal{T}}^{\text{u},L_1}:\mathcal{T}\subset [K], |\mathcal{T}|=t, k\notin \mathcal{T}\right)\]   are requested  by user $k$ and will be sent by the users.	
		Subfiles  \[\left(W_{d_k,\mathcal{T}}^{\text{s}} :\mathcal{T}\subset [K], |\mathcal{T}|=t, k\notin \mathcal{T}\right)\] are requested  by user $k$  and  will be sent by the server.

		
		First consider the subfiles sent by the users. In order to create multicast opportunities among users, we partition the $K$ users into $\alpha$ groups of equal size: \[\mathcal{G}_1,\ldots,\mathcal{G}_{\alpha},\] where for $i,j=1,\ldots,\alpha$,
		$\mathcal{G}_i\subseteq [K]: |\mathcal{G}_i|={\lfloor {K}/{\alpha} \rfloor}$, and  $\mathcal{G}_i\cap \mathcal{G}_j=\emptyset$,~\textnormal{if $i\neq j$}. In each group $\mathcal{G}_i$, one of $\lfloor {K}/{\alpha} \rfloor$ users plays the role of server and  sends symbols based on its cached contents to the remaining $(\lfloor {K}/{\alpha} \rfloor-1)$ users in the group. 
		
		Focus on a group $\mathcal{G}_i$ and a set $\mathcal{S}\subset[K]:|\mathcal{S}|=t+1$.  If $\mathcal{G}_i\subseteq\mathcal{S}$, then all nodes in $\mathcal{G}_i$ share  subfiles 
$$(W^{\text{u},l}_{n,\mathcal{T}}:l\in[L_1],n\in[N],\mathcal{G}_i\subseteq\mathcal{T},|\mathcal{T}|=t).$$
In this case,  user  $k\in\mathcal{G}_i$ sends an XOR symbol that contains the requested subfiles useful for all remaining $\lfloor K/\alpha \rfloor -1$ users in $\mathcal{G}_i$. 
If $\mathcal{S}\subseteq \mathcal{G}_i$, then the nodes in $\mathcal{S}$ share  subfiles  $$(W^{\text{u},l}_{n,\mathcal{T}}:l\in[L_1],n\in[N],\mathcal{T}\subset\mathcal{S},|\mathcal{T}|=t).$$ In this case,  user $k\in\mathcal{S}$ sends an XOR symbol that contains the requested subfiles for all remaining $t$ users in $\mathcal{S}$.  Other groups perform the similar steps and concurrently deliver the remaining requested subfiles to other users.
		
		By changing group partition and performing the delivery strategy described above, we can   send all the requested subfiles 
		\begin{IEEEeqnarray}{rCl}\label{eq_userDfiles}
			(W_{d_k,\mathcal{T}}^{\text{u},1},\ldots,W_{d_k,\mathcal{T}}^{\text{u},L_1}:\mathcal{T}\subset [K], |\mathcal{T}|=t, k\notin \mathcal{T})^K_{k=1}
		\end{IEEEeqnarray}
		to the  users. 
		
		Since $\alpha$ groups  send  signals in a parallel manner ($\alpha$ users  can concurrently deliver contents), and  each user in a group delivers a symbol containing $\min\{ \lfloor K/\alpha \rfloor -1, t\}$ non-repeating pico-files requested by other users, in order to send all requested subfiles in \eqref{eq_userDfiles}, we need  to send in total
		\begin{IEEEeqnarray}{rCl}\label{eq_centraTimes}
			\frac{K\cdot\binom{K-1}{t}\cdot L_1}{\alpha\min\{\lfloor  \frac{K}{\alpha}\rfloor-1,t\}}
		\end{IEEEeqnarray}
		XOR symbols,   each  of  size $ \frac{1-\lambda}{\binom{K}{t} }F$ bits. Notice that  $L_1$ is chosen according to \eqref{eq_Lcondition},  ensuring that   \eqref{eq_centraTimes} equals to an integer. Thus, we obtain $R_1$ as 
		\begin{IEEEeqnarray}{rCl}\label{eq_R2}
			R_2 &=&  \frac{KL_1\cdot\binom{K-1}{t}}{\alpha\min\{\lfloor  \frac{K}{\alpha}\rfloor-1,t\}}\cdot \frac{1-\lambda}{L_1\binom{K}{t} } \nonumber\\
			&=&K\Big(1\!-\!\frac{M}{N}\Big)\frac{1}{1\!+\!t\!+\!{\alpha\min\{\lfloor  \frac{K}{\alpha}\rfloor\!-\!1,t\}}},
		\end{IEEEeqnarray}
	where the last equality holds by \eqref{eq_Acondition}.	
		



		
		Now we describe the delivery of the subfiles sent by the server. Apply the  delivery strategy as in \cite{Centralized}, i.e., the server broadcasts $$\oplus_{k\in \mathcal{S}} W_{d_k,\mathcal{S}\backslash\{k\}}^{\text{s}}
		$$
		for all   $\mathcal{S}\subseteq[K]:|\mathcal{S}|=t+1$.		We obtain  the  transmission rate sent of the server 
		\begin{IEEEeqnarray}{rCl}\label{eq_R1}
		R_1&=&  \lambda \cdot K\left(1-\frac{M}{N}\right)\cdot \frac{1}{1+t}\nonumber\\
&=&	K\Big(1\!-\!\frac{M}{N}\Big)\frac{1}{1\!+\!t\!+\!{\alpha\min\{\lfloor  \frac{K}{\alpha}\rfloor\!-\!1,t\}}}.
\end{IEEEeqnarray}
From \eqref{eq_R2} and \eqref{eq_R1}, we can see that the choice $\lambda$ in    \eqref{eq_Acondition} guarantees equal   communication loads at the server and  users. 
		Since the server and users transmit the signals simultaneously, the transmission delay of the whole network is the maximum between $R_1$ and $R_2$, i.e., $ T_\textnormal{central}=\max\{R_1,R_2\}=\frac{K(1\!-\! {M}/{N})}{1\!+\!t\!+\!{\alpha\min\{\lfloor  {K}/{\alpha}\rfloor\!-\!1,t\}}}$, for $\alpha\in[\alpha_{\max}]$.

		 	\section{Coding Scheme Under Decentralized  Data Placement} \label{Sec_Schemes_Decen}
		In this section, we present the decentralized coded caching scheme with user cooperation where the identities of active users stay unknown to the server and the server has no control over what contents users will prefetch the library during the data placement phase. 		
		To  combine the decentralized coded caching with user cooperation,	there are two main challenges:
	\begin{itemize}
		\item  Given a group partition $\mathcal{G}_1,\ldots,\mathcal{G}_\alpha $, in order to achieve the maximum multicast gain for each group $\mathcal{G}_j$, $\forall j\in[\alpha]$, one user  in  group $\mathcal{G}_j$ should broadcast a coded symbol consisting of $|\mathcal{G}_j|\!-\!1$ useful subfiles required by the remaining users in  group $\mathcal{G}_j$. However, due to the decentralized placement phase, each user uniformly chooses $MF/N$ bits from each file at random, resulting in  subfiles of  distinct sizes  cached by different numbers of users. In this way,  users' demands cannot be satisfied if the group size $\mathcal{G}_j$ is fixed like the centralized scheme.  Thus, the users' partitioning group sizes should change dynamically according to the varying file sizes.
		\item The sizes of partitioning groups should traverse the set $\{2,\ldots,K\}$ as mentioned above, resulting in a  dynamical  cooperation  gain. To achieve the optimal transmission delay, we need to  efficiently allocate the communication loads at the server and the users and fully exploit the multicast gain,  cooperation opportunity among users, and  parallel transmission between the server and users.
\end{itemize}

We first use an illustrative example to show how we group users, split data and delivery data, and then present our generalized deceneralized   coding caching scheme. 
\subsection{An Illustrative Example}
		 Consider a cache-aided network consisting of $K=7$ users. When using the decentralized data placement strategy, the subfiles cached by user $k$ can be written as
		\begin{IEEEeqnarray}{rCl} 
			\Big(W_{n,\mathcal{T}}: n \in [N], k \in \mathcal{T}, \mathcal{T} \subseteq [7]\Big).
		\end{IEEEeqnarray}
		We focus on the delivery of subfiles $W_{n,\mathcal{T}}: n \in [N], k \in \mathcal{T}, |\mathcal{T} |=s=4$, i.e., each subfiles are stored by $s=4$ users.  Similar process can be applied to deliver other subfiles with respect to $s\neq 4$. 
		
		Divide  each subfile into two mini-files $W_{n,\mathcal{T}} = \Big(W_{n,\mathcal{T}}^{\text{s}}, W_{n,\mathcal{T}}^{\text{u}}\Big ) $, where mini-files $\{W_{n,\mathcal{T}}^{\text{s}}\}$ and $\{W_{n,\mathcal{T}}^{\text{u}}\}$ will be sent by the server and users, respectively. To reduce the transmission delay, the size of $W_{n,\mathcal{T}}^{\text{s}}$ and $W_{n,\mathcal{T}}^{\text{u}}$ need to be chosen properly such that $R_1=R_2$, i.e., the transmission rate of the server and users are equal, see   \eqref{eqSizeMini} and \eqref{eq:Ruu} ahead. 
		
		Divide all the users into two non-intersecting groups $(\mathcal{G}_1^{r},\mathcal{G}_2^{r})$, for  $r\in[35]$ which satisfies $$\mathcal{G}_1^{r}\subset[K],\mathcal{G}_2^{r}\subset[K], |\mathcal{G}_1^{r}| = 4, |\mathcal{G}_2^{r}| = 3,  \mathcal{G}_1^{r} \cap \mathcal{G}_2^{r}=\emptyset.$$ There are $\binom{7}{4}=35$ kinds of partitions in total, thus   $r\in [35]$. Note that for any user $k  \in \mathcal{G}^r_i$,  $|\mathcal{G}_i^{r}| -1$ of its requested  mini-files are already cached by the rest users in $\mathcal{G}^r_i$, for $i =1,2$.  
			
	In the delivery phase, one user in each group broadcasts an XOR symbols to all other users in its group, and the two groups 	work in parallel. 	All  mini-files $$(W_{d_k,\mathcal{T}\backslash\{k\}}^{\textnormal{u}}:\mathcal{T}\subseteq[7],k\in[7])$$ are divided into non-overlapping pico-files $W_{d_k,\mathcal{T}\backslash\{k\}}^{\textnormal{u}_1}$ and $W_{d_k,\mathcal{T}\backslash\{k\}}^{\textnormal{u}_2}$, 
	i.e., $W_{d_k,\mathcal{T}\backslash\{k\}}^{\textnormal{u}}=(W_{d_k,\mathcal{T}\backslash\{k\}}^{\textnormal{u}_1},W_{d_k,\mathcal{T}\backslash\{k\}}^{\textnormal{u}_2})$. The size of $W_{n,\mathcal{T}}^{\text{s}}$ and $W_{n,\mathcal{T}}^{\text{u}}$ need to be chosen properly  such that  the transmission rate of group $\mathcal{G}_1^{r}$ and $\mathcal{G}_2^{r}$ are equal, see   \eqref{eqSiseFrag} and \eqref{eqBalance3} ahead.

Split each $W_{d_k,\mathcal{T}\backslash\{k\}}^{\textnormal{u}_1}$ and $W_{d_k,\mathcal{T}\backslash\{k\}}^{\textnormal{u}_2}$  into $3$ and 2 equal fragments, respectively,  
	e.g.,
			\begin{IEEEeqnarray*}{rCl}
		W_{d_2,\{134\}}^{\textnormal{u}_1}&=&\Big(	W_{d_2,\{134\}}^{\textnormal{u}_1,1},W_{d_2,\{134\}}^{\textnormal{u}_1,2},W_{d_2,\{134\}}^{\textnormal{u}_1,3}\Big),\\
		W_{d_2,\{134\}}^{\textnormal{u}_2} &=& \Big(W_{d_2,\{134\}}^{\textnormal{u}_2,1} ,W_{d_2,\{134\}}^{\textnormal{u}_2,2} \Big).
\end{IEEEeqnarray*}
			   In each round, one user in each  group produces and multicasts an XOR symbol to all other users in the same group as shown in Table \ref{tab:in-group}. 
			   
	%

		
		\begin{table}[t]
			\begin{adjustbox}{width=0.7\textwidth,center}
				\begin{threeparttable}
					\caption{Parallel user delivery when $K=7$, $s=4$, $\mathcal{G}_1^{r}=4$ and $\mathcal{G}_2^{r}=3$, $r\in[35]$}
					\label{tab:in-group}
					\begin{tabular}{ll}
						\toprule
						\vspace{-7pt}\\
						\begin{tabular}{c}
						\{1,2,3,4\}\\
							$\text{user}~1\textnormal{: }$ {$W_{d_2,\{134\}}^{\textnormal{u}_1,1}$}$\oplus$ {$W_{d_3,\{124\}}^{\textnormal{u}_1,1}$}$\oplus$ {$W_{d_4,\{123\}}^{\textnormal{u}_1,1}$}
							\vspace{2pt}\\
							$\text{user}~2\textnormal{: }$ {$W_{d_1,\{234\}}^{\textnormal{u}_1,1}$}$\oplus$ {$W_{d_3,\{124\}}^{\textnormal{u}_1,2}$}$\oplus$ {$W_{d_4,\{123\}}^{\textnormal{u}_1,2}$}
							\vspace{2pt}\\
							$\text{user}~3\textnormal{: }$ {$W_{d_2,\{134\}}^{\textnormal{u}_1,2}$}$\oplus$ {$W_{d_1,\{234\}}^{\textnormal{u}_1,2}$}$\oplus$ {$W_{d_4,\{123\}}^{\textnormal{u}_1,3}$}
							\vspace{2pt}\\
							$\hspace{-0.7ex}\text{user}~4\textnormal{: }$ {$W_{d_2,\{134\}}^{\textnormal{u}_1,3}$}$\oplus$ {$W_{d_3,\{124\}}^{\textnormal{u}_1,3}$}$\oplus$ {$W_{d_1,\{234\}}^{\textnormal{u}_1,3}$}
						\end{tabular}
						& \hspace{-17pt}
						\begin{tabular}{c}
						\{5,6,7\}\\
							$\text{user}~5\textnormal{: }\underset{x\in\{1234\}}{\cup}$ {$W_{d_6,\{57x\}}^{\textnormal{u}_2,1}$}$\oplus$ {$W_{d_7,\{56x\}}^{\textnormal{u}_2,1}$}
							\vspace{6pt}\\
							$\text{user}~6\textnormal{: }\underset{x\in\{1234\}}{\cup}${$W_{d_5,\{67x\}}^{\textnormal{u}_2,1}$}$\oplus$ {$W_{d_7,\{56x\}}^{\textnormal{u}_2,2}$}
							\vspace{6pt}\\
							$\hspace{-0.7ex}\text{user}~7\textnormal{: }\underset{x\in\{1234\}}{\cup}$ {$W_{d_6,\{57x\}}^{\textnormal{u}_2,2}$}$\oplus${$W_{d_5,\{67x\}}^{\textnormal{u}_2,2}$}
						\end{tabular}\\
						\vspace{-7pt}\\
						\midrule
						\vspace{-7pt}\\
						\begin{tabular}{c}
						\{1,2,3,5\}\\
							$$$\text{user}~1\textnormal{: }${$ W_{d_2,\{135\}}^{\textnormal{u}_1,1}$}$\oplus$\highlight[fill opacity=0]{$W_{d_3,\{125\}}^{\textnormal{u}_1,1}$}$\oplus$\highlight[fill opacity=0]{$W_{d_5,\{123\}}^{\textnormal{u}_1,1}$}
							\vspace{2pt}\\
							$\text{user}~2\textnormal{: }${$ W_{d_1,\{235\}}^{\textnormal{u}_1,1}$}$\oplus$\highlight[fill opacity=0]{$W_{d_3,\{125\}}^{\textnormal{u}_1,2}$}$\oplus$\highlight[fill opacity=0]{$ W_{d_5,\{123\}}^{\textnormal{u}_1,2}$}
							\vspace{2pt}\\
							$\text{user}~3\textnormal{: }${$ W_{d_2,\{135\}}^{\textnormal{u}_1,2}$}$\oplus${$ W_{d_1,\{235\}}^{\textnormal{u}_1,2}$}$\oplus${$ W_{d_5,\{123\}}^{\textnormal{u}_1,3}$}
							\vspace{2pt}\\
							$\hspace{-0.7ex}\text{user}~5\textnormal{: }${$ W_{d_2,\{135\}}^{\textnormal{u}_1,3}$}$\oplus${$W_{d_3,\{125\}}^{\textnormal{u}_1,3}$}$\oplus${$W_{d_1,\{235\}}^{\textnormal{u}_1,3}$}
						\end{tabular}
						& \hspace{-17pt}
						\begin{tabular}{c}
						\{4,6,7\}\\
							$\text{user}~4\textnormal{: }\underset{x\in\{1235\}}{\cup}${$W_{d_6,\{47x\}}^{\textnormal{u}_2,y_{(..)}}$}$\oplus${$W_{d_7,\{46x\}}^{\textnormal{u}_2,y_{(..)}}$}
							\vspace{6pt}\\
							$\text{user}~6\textnormal{: }\underset{x\in\{1235\}}{\cup}${$W_{d_4,\{67x\}}^{\textnormal{u}_2,1}$}$\oplus${$W_{d_7,\{46x\}}^{\textnormal{u}_2,y_{(..)}}$}
							\vspace{6pt}\\
							$\hspace{-0.7ex}\text{user}~7\textnormal{: }\underset{x\in\{1235\}}{\cup}${$W_{d_6,\{47x\}}^{\textnormal{u}_2,y_{(..)}}$}$\oplus$\highlight[fill opacity=0]{$W_{d_4,\{67x\}}^{\textnormal{u}_2,2}$}
						\end{tabular}\\
						\vspace{-7pt}\\
						\midrule
						\vspace{-7pt}\\
						\begin{tabular}{c}
						\{1,2,3,6\}\\
							$\text{user}~1\textnormal{: }$\highlight[fill opacity=0]{$ W_{d_2,\{136\}}^{\textnormal{u}_1,1}$}$\oplus$\highlight[fill opacity=0]{$ W_{d_3,\{126\}}^{\textnormal{u}_1,1}$}$\oplus$\highlight[fill opacity=0]{$ W_{d_6,\{123\}}^{\textnormal{u}_1,1}$}
							\vspace{6pt}\\
							$\text{user}~2\textnormal{: }$\highlight[fill opacity=0]{$ W_{d_1,\{236\}}^{\textnormal{u}_1,1}$}$\oplus$\highlight[fill opacity=0]{$ W_{d_3,\{126\}}^{\textnormal{u}_1,2}$}$\oplus$\highlight[fill opacity=0]{$ W_{d_6,\{123\}}^{\textnormal{u}_1,2}$}
							\vspace{6pt}\\
							$\text{user}~3\textnormal{: }$\highlight[fill opacity=0]{$ W_{d_2,\{136\}}^{\textnormal{u}_1,2}$}$\oplus$\highlight[fill opacity=0]{$ W_{d_1,\{236\}}^{\textnormal{u}_1,2}$}$\oplus$\highlight[fill opacity=0]{$ W_{d_6,\{123\}}^{\textnormal{u}_1,3}$}
							\vspace{6pt}\\
							$\hspace{-0.7ex}\text{user}~6\textnormal{: }$\highlight[fill opacity=0]{$  W_{d_2,\{136\}}^{\textnormal{u}_1,3}$}$\oplus$\highlight[fill opacity=0]{$ W_{d_3,\{126\}}^{\textnormal{u}_1,3}$}$\oplus$\highlight[fill opacity=0]{$ W_{d_1,\{236\}}^{\textnormal{u}_1,3}$}
						\end{tabular}
						& \hspace{-17pt}
						\begin{tabular}{c}
						\{4,5,7\}\\
							$\text{user}~4\textnormal{: }\underset{x\in\{1236\}}{\cup}${$ W_{d_5,\{47x\}}^{\textnormal{u}_2,y_{(..)}}$}$\oplus${$ W_{d_7,\{45x\}}^{\textnormal{u}_2,y_{(..)}}$}
							\vspace{6pt}\\
							$\text{user}~5\textnormal{: } \underset{x\in\{1236\}}{\cup}$\highlight[fill opacity=0]{$ W_{d_4,\{57x\}}^{\textnormal{u}_2,1}$}$\oplus${$ W_{d_7,\{45x\}}^{\textnormal{u}_2,y_{(..)}}$}
							\vspace{6pt}\\
							$\hspace{-0.7ex}\text{user}~7\textnormal{: }\underset{x\in\{1236\}}{\cup}${$ W_{d_5,\{47x\}}^{\textnormal{u}_2,y_{(..)}}$}$\oplus${$ W_{d_4,\{57x\}}^{\textnormal{u}_2,2}$}
						\end{tabular}\\
						\vspace{-7pt}\\
						
						\midrule
						\vspace{-11pt}\\
						$\hspace{20pt}\cdots$ \hspace{60pt}$\cdots\cdots$
						& \hspace{10pt} $\cdots$ \hspace{50pt}$\cdots\cdots$
						\\
						\bottomrule
					\end{tabular}
					\begin{tablenotes}
						\normalsize
						\item There should be 35 partitions in total while the table only showed 3 partitions. 
					\end{tablenotes}
				\end{threeparttable}
			\end{adjustbox}
		\end{table}

			Note that in this example, each group only appears once in all partitions.  For larger $K$, each group could appear multiple times in different partitions, which poses more difficulties in   the data transmission such that no  fragment is repeatedly sent.
		
		 \subsection{The Generalized Decentralized Coded Caching Scheme}\label{SecDecScheme}
	In the  placement phase, each user $k$ applies the caching function to map a subset of $\frac{MF}{N}$ bits of file $W_n, n=1,...,N,$ into its  cache memory at random:
		\begin{IEEEeqnarray}{rCl}\label{eq:Subfiles}
			W_n = \Big(  W_{n,\mathcal{T}} : \mathcal{T} \subseteq [K] \Big).
		\end{IEEEeqnarray}
		The subfiles cached by user $k$ can be written as
		\begin{IEEEeqnarray}{rCl}\label{eq:Dcache}
			\Big(W_{n,\mathcal{T}}: n \in [N], k \in \mathcal{T}, \mathcal{T} \subseteq [K]\Big).
		\end{IEEEeqnarray}
		When the file size $F$ is sufficiently large, by the law of large numbers, the subfile size with high probability can be written by
		\begin{IEEEeqnarray}{rCl}\label{eq:DFileSize}
			|W_{n,\mathcal{T}}| &\approx& \Big(\frac{M}{N}\Big)^{|\mathcal{T}|}\Big(1-\frac{M}{N}\Big)^{K-|\mathcal{T}|}\nonumber\\
			&=&p^{|\mathcal{T}|}(1-p)^{K-|\mathcal{T}|} .
		\end{IEEEeqnarray}



		
		The delivery procedure can be characterized on three different levels: allocation between the server's and user's communication loads,  inner-group coding (i.e., transmission in each group) and parallel delivery among groups.
		\subsubsection{Allocation between the server's and user's communication loads}
		{ 	
		Split each subfile $W_{n,\mathcal{T}}$, for $\mathcal{T}\subseteq [K]: \mathcal{T}  \neq \emptyset$,    into two non-overlapping  mini-files 
		\[W_{n,\mathcal{T}} = \Big(W_{n,\mathcal{T}}^{\text{s}}, W_{n,\mathcal{T}}^{\text{u}}\Big ),\]
		where  
			\begin{equation}\label{eqSizeMini}
		\begin{aligned}
		|W_{n,\mathcal{T}}^\textnormal{s}| = \lambda \cdot |W_{n,\mathcal{T}}|,\\
		|W_{n,\mathcal{T}}^\textnormal{u}| = (1-\lambda) \cdot |W_{n,\mathcal{T}}|,
		\end{aligned}
		\end{equation}
		and $\lambda$ is a design parameter. 
		
		   Mini-files 	$(W_{d_k,\mathcal{T}\backslash\{k\}}^{\text{s}}:k\in[K])$ are   to be sent by the server using the original decentralized coded caching scheme \cite{Decentralized}. The corresponding transmission delay   is 
\begin{IEEEeqnarray}{rCl}\label{eq:Rss}
\lambda R_\textnormal{s}=\lambda \frac{1-M/N}{M/N}\Big(1-\big(1-\frac{M}{N}\big)^K\Big),
\end{IEEEeqnarray}
where $R_\textnormal{s}$ coincides with the definition  in \eqref{eq:Rs}.

  Mini-files $(W_{d_k,\mathcal{T}\backslash\{k\}}^{\text{u}}:k\in[K])$ are  to be sent by   users using \emph{parallel user delivery} descrbied in Section \ref{SubSec:Inter-G_B}. The corresponding transmission rate is 
  \begin{IEEEeqnarray}{rCl}\label{eq:Ruu}
	R_1    =   (1-\lambda)R_\textnormal{u}, 
\end{IEEEeqnarray}
where $R_\textnormal{u}$ is transmission bits normalized by $F$ sent in the cooperation network.

	  Since subfile $W_{d_k,\emptyset}$ is not cached by any user and  must be  sent exclusively from the server, the corresponding transmission delay for sending  $(W_{d_k,\emptyset}:k\in[K])$ is
	  \begin{IEEEeqnarray}{rCl}\label{eq:Rempy}
	  R_\emptyset= K\big(1-\frac{M}{N}\big)^K,
\end{IEEEeqnarray}
 where $R_\emptyset$  coincides with the definition in  \eqref{eq:Re}. 
	 
	By \eqref{eq:Rss}, \eqref{eq:Ruu} and \eqref{eq:Rempy}, we have
  \begin{IEEEeqnarray}{rCl}
			R_1 && = R_\emptyset +  \lambda R_\textnormal{s}, \nonumber \\
			R_2 && =   (1-\lambda)R_\textnormal{u}.
		\end{IEEEeqnarray} 
	 
	According to \eqref{eq:fDuplex},   we have $T_\textnormal{decentral} = \max\{R_1,R_2\}$. The parameter $\lambda$ is chosen such that   $T_\textnormal{decentral}$ is minimized.

		
	\begin{Remark}[{Choice of $\lambda$}] 	If $R_\textnormal{u} < R_\emptyset$, then the inequality $R_2 \leq R_1$ always holds. In this case, only when $\lambda = 0$,  $T_\textnormal{decentral}$ reaches the minimum $$T_\textnormal{decentral} = R_\emptyset.$$ 
		If $R_\textnormal{u} \geq R_\emptyset$, solving $R_1 = R_2$ yields $\lambda= \frac{R_\textnormal{u}-R_\emptyset}{R_\textnormal{s}+R_\textnormal{u}}$ and 
		$$T_\textnormal{decentral} = \frac{R_\textnormal{s}R_\textnormal{u}}{R_\textnormal{s}+R_\textnormal{u}-R_\emptyset}.$$
\end{Remark}

		\subsubsection{Inner-group  coding} \label{SubSec:GroupCoding}
		Given parameters $(s,\mathcal{G}, \textnormal{p},\gamma)$ where $s\in[K-1]$, $\mathcal{G}\subseteq[K]$, $\textnormal{p}\in\{\textnormal{u},\textnormal{u}_1,\textnormal{u}_2\}$ with indicators $\textnormal{u},\textnormal{u}_1,\textnormal{u}_2$ described later in Section \ref{SubSec:Inter-G_B}  and $\gamma\in\mathbb{Z}^+$, we present  how to successfully deliver   
\[		(W^\text{p}_{d_k,\mathcal{S}\backslash\{k\}}:\forall \mathcal{S} \subseteq [K],  |\mathcal{S}| = s, \mathcal{G}\subseteq\mathcal{S})\]
to  every user $k\in\mathcal{G}$  through user cooperation.

 Split	each $W_{d_k,\mathcal{S}\backslash\{k\}}^{\textnormal{p}}$ into $(|\mathcal{G}| -1)\gamma$  non-overlapping fragments of equal size, i.e.,
			\begin{IEEEeqnarray}{rCl}\label{eq_split}
				W_{d_k,\mathcal{S} \backslash\{k\}}^{\textnormal{p}} = \Big( W_{d_k,\mathcal{S} \backslash\{k\}}^{\textnormal{p},l}: l \in [(|\mathcal{G}| -1)\gamma] \Big),
			\end{IEEEeqnarray}  
			and  each user $k \in \mathcal{G}$ takes turn to broadcast XOR symbols
		\begin{IEEEeqnarray}{rCl} \label{eqxorSymbol}
			X_{k,\mathcal{G},s}^{{\text{p}}} \triangleq &&  \oplus_{j \in \mathcal{G}\backslash\{k\}}
			W_{d_j,\mathcal{S} \backslash\{j\}}^{{\text{p}}, l{(j,\mathcal{G},\mathcal{S})}}, 
		\end{IEEEeqnarray}
where $l{(k,  \mathcal{G},\mathcal{S})}\in  [(|\mathcal{G}| -1)\gamma] $ is a  function of {$(k,  \mathcal{G},\mathcal{S})$} which avoids redundant transmission of any fragments. 
	The XOR symbol $X_{k,\mathcal{G},s}^{{\text{p}}} $ will be received and decoded by the remaining users in $\mathcal{G}$.		 

For each group  $\mathcal{G}$, inner-group coding will recover  in total ${K-|\mathcal{G}|\choose s-|\mathcal{G}|}$ of  $W^\text{p}_{d_k,\mathcal{S}\backslash\{k\}}$, and  each XOR symbol $X_{k,\mathcal{G},s}^{{\text{p}}}$ in \eqref{eqxorSymbol} contains fragments required by  $|\mathcal{G}|-1$ users in $\mathcal{G}$. 
	

		\subsubsection{Parallel  delivery among groups}\label{SubSec:Inter-G_B}
In order to provide parallel delivery among groups, as well as  avoid redundant transmission of any content among all groups, we need to carefully design how to partition groups and how  signals are transmitted among groups.		

{ The parallel user delivery consists of $(K-1)$ rounds characterized by $s=2,\dots, K$. In each round,  mini-files 
		\[		(W^\text{u}_{d_k,\mathcal{T}\backslash\{k\}}:\forall \mathcal{T} \subseteq [K],  |\mathcal{T}| = s,k\in[K])\]
		are recovered through user cooperation.  
	Based on $K$, $s$ and $\alpha_{\max}$, i.e., the maximum number of users allowed for parallel transmission,   the delivery strategy of the users is divided into 3 cases:
		\begin{itemize}
		 
		\item Case 1: $\lceil \frac{K}{s} \rceil > \alpha_{\max}$. In this case,  $\alpha_{\max}$ users are allowed to send data simultaneously. Select $s\cdot\alpha_{\max}$ users from all users and divide them into $\alpha_{\max}$ groups of equal size $s$.  
			The total number of  such kind  partition is
			\begin{IEEEeqnarray}{rCl}\label{eqBeta2}
			\beta_1\triangleq \frac{{K \choose s}{K-s \choose s}\cdots{K-s (\alpha_{\max} -1)\choose s }}{\alpha_{\max}!} .
\end{IEEEeqnarray}
				 In each partition,  each user is selected from $ \alpha_{\max}$ groups  an individual group and  sends data in parallel.
		\item Case 2:  $\lceil \frac{K}{s} \rceil \leq \alpha_{\max}$ and $(K~\text{mod}~ s) < 2$.  In this case,   	every $s$ users form a group. Choose $\lfloor \frac{K}{s} \rfloor s$ users from all users and partition them into $\lfloor \frac{K}{s} \rfloor$  groups of equal size $s$.  
			The total number of  such kind  partition is
			\begin{IEEEeqnarray}{rCl}\label{eqBeta1}
			\beta_2\triangleq \frac{{K \choose s}{K-s \choose s}\cdots{K-s(\lfloor \frac{K}{s} \rfloor-1) \choose s }}{{\lfloor \frac{K}{s} \rfloor}!} .
\end{IEEEeqnarray}
			 In each partition,  each user is selected from $ \lfloor \frac{K}{s} \rfloor$ groups  an individual group and  sends data in parallel

			\item Case 3:  $\lceil \frac{K}{s} \rceil \leq \alpha_{\max}$ and $(K~\text{mod}~s) \geq 2$. In this case, every $s$ users form a group, resulting in $\lfloor \frac{K}{s} \rfloor$ groups consisting of $s\lfloor \frac{K}{s} \rfloor$ users. The  remaining $(K~\text{mod}~ s)$ users forms another group. The total number of  such kind  partition is
			\begin{IEEEeqnarray}{rCl}\label{eqBeta3}
			\beta_3=\beta_2.
\end{IEEEeqnarray}
	 In each partition,  each user is selected from  $ \lceil \frac{K}{s} \rceil$ groups  an individual group and  sends data in parallel

		\end{itemize}
		
		Thus the exact number of users who parallelly send signals can be written as follows:
	 
		\begin{equation} \label{eq:optimalAlphaD}
		\alpha_\textnormal{D}=\!\left\{
		\begin{aligned}
		&\alpha_{\max},   \ &\text{case 1}, \\
		&\lfloor \frac{K}{s} \rfloor,  \  &\text{case 2},\\ 
				&\lceil \frac{K}{s} \rceil,   \  &\text{case 3}.\\
		\end{aligned}
		\right.
		\end{equation}
	Note that for case $c\in\{1,2,3\}$,  each group $\mathcal{G}$ among  $[\beta_c]$  partitions  re-appears 
\begin{IEEEeqnarray}{rCl}\label{eqNG2}
N_{\mathcal{G}}\triangleq\frac{{K-s \choose s}\cdots{K-s\cdot(\alpha_\textnormal{D} -1) \choose s }}{(\alpha_\textnormal{D}-1)!}
\end{IEEEeqnarray}
 times. 
 
 Now we present our decentralized scheme for these three cases in details.

\emph{Case 1} ($\lceil \frac{K}{s} \rceil > \alpha_{\max}$):	
Consider a partition  $r\in[\beta_1]$, denoted as  
 \[\mathcal{G}^r_1,\ldots,\mathcal{G}^r_{\alpha_\textnormal{D}},\] 	where  $|\mathcal{G}^r_i| = s$ and $\mathcal{G}^r_i\cap\mathcal{G}^r_j=\emptyset$,  $\forall i,j\in [ \alpha_\textnormal{D}]$ and $i\neq j$.  
 
 
  Since each group $\mathcal{G}^r_i$  re-appears  $N_{\mathcal{G}^r_i}$ times among $[\beta_1]$ partitions, and  $(|\mathcal{G}^r_i|-1)$ users take turns to broadcast XOR symbol \eqref{eqxorSymbol} in  each group $\mathcal{G}^r_i$,
  in order  to   guarantee that each group sends unique  fragments  without repetition, we split each mini-file $W^\text{u}_{d_k,\mathcal{S}\backslash\{k\}}$ into $(|\mathcal{G}^r_i|-1) N_{\mathcal{G}^r_i}$ fragments  of equal size.   

Groups $ \mathcal{G}^r_i$, $r\in[\beta_1]$ and $i\in[\alpha_\textnormal{D}]$,  performs inner group coding (see Section \eqref{SubSec:GroupCoding}) with parameters  
$$(s,\mathcal{G}_i^r, \textnormal{u},  N_{\mathcal{G}^r_i}),$$
 for all $s$ satisfying $\lceil \frac{K}{s} \rceil > \alpha_{\max}$. For each round $r$, all groups $ \mathcal{G}^r_{1},\ldots, \mathcal{G}^r_{\alpha_\textnormal{D}}$ parallelly send XOR symbols  containing 
$|\mathcal{G}^r_i|-1$ fragments required by other users of its group. By the fact that the partitioned  groups  traverse every set $\mathcal{T}$, i.e.,
\[\mathcal{T} \subseteq  \{ \mathcal{G}^r_1\cup\ldots\cup\mathcal{G}^r_{\alpha_\textnormal{D}}\}^{\beta_1}_{r=1},\forall \mathcal{T}\subseteq[K]:|\mathcal{T}|=s,\]
 and since inner group coding  enables  each group $\mathcal{G}^r_i$ to recover 
 \[		(W^\text{u}_{d_k,\mathcal{S}\backslash\{k\}}:\forall \mathcal{S} \subseteq [K],  |\mathcal{S}| = s, \mathcal{G}^r_i\subseteq\mathcal{S},k\in[K]),\] 
we are able to recover all  required mini-files
\[		(W^\text{u}_{d_k,\mathcal{T}\backslash\{k\}}:\forall \mathcal{T} \subseteq [K],  |\mathcal{T}| = s, k\in[K]).\]
The transmission delay of case 1 at round $s$ is 			thus
\begin{IEEEeqnarray}{rCl}\label{eqRateCase1}
			R_\text{case1}^{{\text{u}}}(s) && \triangleq  \sum_{r\in[\beta_1]} \sum_{k\in \mathcal{G}^r_i}|X_{k,\mathcal{G}^r_i,s}^{{\text{u}}}| \nonumber \\
			&& = \beta_1 |\mathcal{G}^r_i| \frac{ |W^\text{u}_{d_k,\mathcal{T}\backslash\{k\}}|}{(|\mathcal{G}^r_i|-1) N_{\mathcal{G}^r_i}}   \nonumber  \\
			&&\stackrel{(\text{a})}{=} \frac{K {K-1\choose s-1}}{\alpha_\textnormal{D}(s-1)}   |W^\text{u}_{d_k,\mathcal{T}\backslash\{k\}}|  \nonumber\\
				&& = \frac{K {K-1\choose s-1}}{\alpha_{\max}(s-1)}  (1-\lambda)p^{s-1}(1-p)^{K-s+1},
\end{IEEEeqnarray}
where (a) follows by \eqref{eqNG2}.

 \emph{Case 2} ($\lceil \frac{K}{s} \rceil \leq \alpha_{\max}$ and $(K~\text{mod}~ s) < 2$):
We apply the same delivery procedure as case 1, except that  $\beta_1$ is replaced by $\beta_2$ and $\alpha_\textnormal{D}=\lfloor \frac{K}{s} \rfloor$, and   obtain transmission delay of each round $s$:
			\begin{IEEEeqnarray}{rCl}\label{eqRateCase2}
			R_\text{case2}^{{\text{u}}}(s) &&  = \frac{K {K-1\choose s-1}}{\alpha_{D}(s-1)}   |W^\text{u}_{d_k,\mathcal{T}\backslash\{k\}}|  \nonumber\\
				&& = \frac{K {K-1\choose s-1}}{\lfloor \frac{K}{s} \rfloor(s-1)  }  (1-\lambda)p^{s-1}(1-p)^{K-s+1}.
				\end{IEEEeqnarray}

 \emph{Case 3} ( $\lceil \frac{K}{s} \rceil \leq \alpha_{\max}$ and $(K~\text{mod}~s) \geq 2$):		
		Consider a partition  $r\in[\beta_3]$, denoted as  
 \[\mathcal{G}^r_1,\ldots,\mathcal{G}^r_{\alpha_\textnormal{D}},\]	where   
  $\mathcal{G}^r_i\subseteq[K]$,  $\mathcal{G}^r_i\cap\mathcal{G}^r_j=\emptyset$,   $\forall i,j\in [ {\alpha_\textnormal{D}-1}]$ and $i\neq j$ and  $\mathcal{G}^r_{\alpha_\textnormal{D}} = [K]\backslash (\mathcal{G}_1,\ldots, \mathcal{G}^r_{{\alpha_\textnormal{D}-1}})$  with 
    $$|\mathcal{G}^r_i| = s,  ~|\mathcal{G}^r_{\alpha_\textnormal{D}}| =K~\text{mod} ~s.$$    
    
	Since 	group $\mathcal{G}^r_i: i\in [ \alpha_\textnormal{D}-1]$ and $\mathcal{G}^r_{\alpha_\textnormal{D}}$ have different size, we further  split each mini-file $W_{d_k,\mathcal{T}\backslash\{k\}}^{\text{u}}$ into 2 non-overlapping fragments such that 
		\begin{IEEEeqnarray}{rCl}\label{eqSiseFrag}
			|W_{d_k,\mathcal{T}\backslash\{k\}}^{\textnormal{u}_1}|  = \lambda_2|W_{d_k,\mathcal{T}\backslash\{k\}}^\textnormal{u}|,\\
			|W_{d_k,\mathcal{T}\backslash\{k\}}^{\textnormal{u}_2}| = (1-\lambda_2)|W_{d_k,\mathcal{T}\backslash\{k\}}^\textnormal{u}|, \nonumber
		\end{IEEEeqnarray}
		where $\lambda_2   \in [0,1]$ is a designed parameter which should satisfy \eqref{eqBalance3}.


	Split each mini-files $W^{\text{u}_1}_{d_k,\mathcal{S}\backslash\{k\}}$ and $W^{\text{u}_2}_{d_k,\mathcal{S}\backslash\{k\}}$ into  fragments of equal size:
 \begin{IEEEeqnarray*}{rCl}
				W_{d_k,\mathcal{S} \backslash\{k\}}^{\textnormal{u}_1} &=& \Big( W_{d_k,\mathcal{S} \backslash\{k\}}^{{\textnormal{u}_1},l}: l \in  [(s-1)N_{\mathcal{G}^r_i}] \Big),\\
				W_{d_k,\mathcal{S} \backslash\{k\}}^{\textnormal{u}_2} &=& \Big( W_{d_k,\mathcal{S} \backslash\{k\}}^{{\textnormal{u}_2},l}:\nonumber\\
				&&\hspace{1ex} l \in  \left[\big(|\mathcal{G}^r_{\alpha_\textnormal{D}}| -1\big){s-1 \choose |\mathcal{G}^r_{\alpha_\textnormal{D}}| -1}N_{\mathcal{G}^r_{i}}\right] \Big).
			\end{IEEEeqnarray*} 
			
	Following the similar encoding operation in \eqref{eqxorSymbol},		 group $\mathcal{G}^r_{i}:i \in [\alpha_\textnormal{D}-1]$ and group $\mathcal{G}^r_{\alpha_\textnormal{D}}$ send the following XOR symbols respectively:
		\begin{IEEEeqnarray}{rCl}
			\nonumber
			&&\big(X_{k,\mathcal{G}^r_i,s}^{\textnormal{u}_1}: k \in \mathcal{G}^r_i\big)_{i=1}^{(\alpha_\textnormal{D}-1)}, \\  &&\big(X_{k,\mathcal{G}^r_{\alpha_\textnormal{D}},s}^{\textnormal{u}_2}: k \in \mathcal{G}^r_{\alpha_\textnormal{D}}\big).\nonumber
		\end{IEEEeqnarray}
		For each  $s\in \{2,\ldots,K\}$, the  transmission delay for sending
XOR symbols above by  group $\mathcal{G}^r_{i}:i \in [\alpha_\textnormal{D}-1]$ and group $\mathcal{G}^r_{\lceil \frac{K}{s}\rceil}$ can be written as		
			\begin{IEEEeqnarray}{rCl}
				R_\text{case3}^{{\text{u}_1}}(s) = \frac{\lambda_2K{K-1\choose s-1}}{(\alpha_\textnormal{D}-1)(s-1)}\cdot   |W^\text{u}_{d_k,\mathcal{T}\backslash\{k\}}|, \nonumber \\
				R_\text{case3}^{{\text{u}_2}}(s)  = \frac{(1-\lambda_2)K{K-1\choose s-1}}{(K ~\text{mod}~ s)-1}\cdot   |W^\text{u}_{d_k,\mathcal{T}\backslash\{k\}}|, \nonumber
			\end{IEEEeqnarray}
respectively. Since $\mathcal{G}_{i}:i \in [\lfloor \frac{K}{s} \rfloor]$ and group $\mathcal{G}_{\lceil \frac{K}{s}\rceil}$ can send signals in parallel, by letting 
\begin{IEEEeqnarray}{rCl}\label{eqBalance3}
R_\text{case3}^{{\text{u}_1}}(s)  = R_\text{case3}^{{\text{u}_2}}(s),
\end{IEEEeqnarray}
we eliminate the parameter $\lambda_2$ and obtain the balanced transmission delay at users for case 3:
			\begin{IEEEeqnarray}{rCl}\label{eqRateCase3}
				R_\text{case3}^\text{u}(s) &\triangleq &\frac{K{K-1 \choose s-1}}{K-1-\lfloor\frac{K}{s}\rfloor}   (1-\lambda)p^{s-1}(1-p)^{K-s+1}.  ~
			\end{IEEEeqnarray} 
			
\begin{Remark}\label{eqRemarkComputeDelay}
The condition $\lceil \frac{K}{s} \rceil > \alpha_{\max}$ in Case 1 implies that $s\leq {\lceil \frac{K}{\alpha_{\max}}\rceil-1}$.  In this regime, scheme of Case 1 is working  and the delay is given in \eqref{eqRateCase2}.   If $s\geq {\lceil \frac{K}{\alpha_{\max}}\rceil-1}$ and $(K \mod s)<2 $,  scheme in Case 2 starts to work and the delay is given in \eqref{eqRateCase2}; If $s\geq {\lceil \frac{K}{\alpha_{\max}}\rceil-1}$ and $(K \mod s)\geq 2 $, scheme in Case 3 starts to work  and the delay is given in \eqref{eqRateCase3}.
\end{Remark}
For each round  $s\in\{2,\ldots,K\}$,   all requested  mini-files can be recovered by the delivery strategies above. 
By Remark \ref{eqRemarkComputeDelay}, the achievable delay caused by users' transmission is 
		\begin{IEEEeqnarray}{rCl}\label{eq:pf_Ru}
			R_\textnormal{2} 
			 && =   (1-\lambda)   \frac{1}{\alpha_\textnormal{max}}  \sum_{s=2}^{\lceil \frac{K}{\alpha_\textnormal{max}}\rceil-1}\frac{s{K \choose s}}{s-1}p^{s-1}(1-p)^{K-s+1}   \nonumber \\ &&\hspace{6pt}+(1-\lambda)  \sum_{s= \lceil \frac{K}{\alpha_\textnormal{max}}\rceil}^{K}     \frac{K {K-1 \choose s-1}}{f(K,s)} p^{s-1}(1-p)^{K-s+1},\\
			&& = (1-\lambda)  R_\text{u},
		\end{IEEEeqnarray}
		where $R_\text{u}$ is defined in \eqref{eq:Ru} and \begin{IEEEeqnarray}{rCl} 
				f(K,s) \triangleq
				\left\{
				\begin{aligned}
					&\lfloor \frac{K}{s} \rfloor(s-1), & (K~\textnormal{mod}~s)<2,\\
					&K-1-\lfloor{K}/{s}\rfloor, &(K~\textnormal{mod}~s)\geq2.
				\end{aligned}
				\right.
			\end{IEEEeqnarray}

	For completeness,  we formally describe the
	procedures of user-server tradeoff, inner-group coding and parallel user delivery, for a network with $N$ files and  $K$ users in Algorithm 1.
	\begin{algorithm}[]  \begin{spacing}{1} 
		\caption{Delivery Phase in The Decentralized  Scheme}
		\label{Decen_Algo}
		\begin{algorithmic}
		\State	$\mathbf{d}\xleftarrow{}(d_1,\ldots,d_K)$,  $\mathcal{T}\xleftarrow{}\{\mathcal{T}\subset[K]: \mathcal{T}\neq \emptyset\}$ 
		  \State 			 $(R_\emptyset,R_\textnormal{s},R_\textnormal{u})\xleftarrow{}$ transmission delay defined in \eqref{eq:achievable_rate_decen_whole}  
		  \State $(\beta_1,\beta_2,\beta_3)\xleftarrow{}$ integers defined in (\ref{eqBeta2}-\ref{eqBeta3}), $N_{\mathcal{G}} \xleftarrow{}$ integers defined in (\ref{eqNG2})

		\Procedure{User-Server tradeoff~}{$R_\emptyset,R_\textnormal{s},R_\textnormal{u}$}
			\State	  $W_{d_k,\mathcal{T}}\xrightarrow{\text{split}}\Big(W_{d_k,\mathcal{T}}^{\text{s}}, W_{d_k,\mathcal{T}}^{\text{u}}\Big )$, $\forall k, \mathcal{T}$, with\\ \quad~~ $|W_{d_k,\mathcal{T}}^{\text{s}}|=\lambda |W_{d_k,\mathcal{T}}| $,  $|W_{d_k,\mathcal{T}}^{\text{u}}|=(1-\lambda) |W_{d_k,\mathcal{T}}|$ 
			\If{$R_\textnormal{u}\leq R_\emptyset$}
			\State $\mathcal{W}_{\mathbf{d}}\xleftarrow{}\{W_{d_k,\mathcal{T}}:\forall k, \mathcal{T}\}$
			\State  \hspace{-6pt}$\left.\begin{array}{l}
			\textnormal{Server sequentially sends}~W_{d_k,\mathcal{\emptyset}},\forall k\\
			 \textnormal{PARALLEL USER DELIVERY~}(\mathcal{W}_{\mathbf{d}})
			\end{array}\hspace{-5pt}\right\}$ parallel
			\Else
			\State $\lambda\leftarrow(R_\textnormal{u}-R_\emptyset )/(R_\textnormal{s}+R_\textnormal{u})$
						\State $\mathcal{W}^\textnormal{u}_{\mathbf{d}}\xleftarrow{}\{W^\textnormal{u}_{d_k,\mathcal{T}}:\forall k, \mathcal{T}\}$,  $\mathcal{W}^\textnormal{s}_{\mathbf{d}}\xleftarrow{}\{W^\textnormal{s}_{d_k,\mathcal{T}}:\forall k, \mathcal{T}\}$
			
			\State \hspace{-6pt}$\left.\begin{array}{l}
			\textnormal{Server  sequentially sends}~  W_{d_k,\mathcal{\emptyset}},\forall k   
			~ \textnormal{then sends}  ~ \mathcal{W}^\textnormal{s}_{\mathbf{d}} ~ \textnormal{using scheme in \cite{Decentralized}} \\
			 \textnormal{PARALLEL USER DELIVERY~}(\mathcal{W}^\textnormal{u}_{\mathbf{d}})
			\end{array}\hspace{-5pt}\right\}$ parallel

			\EndIf
			 
			
			\EndProcedure
			
			\Procedure{Parallel User Delivery~}{$\mathcal{W}^\textnormal{u}_{\mathbf{d}}$}
			\For{$s\in\{2,\ldots,K-1\}$}
			\State $s^{*} \xleftarrow{} \big(K~\text{mod}~ s\big)$
		 \If{$\lceil \frac{K}{s} \rceil > \alpha_{\max}$}
			  \For{$r\in[\beta_1]$}
			\State Partition: $\{\mathcal{G}^r_1,..,\mathcal{G}^r_{\alpha_{\max}}\}$,  $|\mathcal{G}^r_i|=s,{i\leq \alpha_{\max}}$
			\State For  $i\in [ \alpha_{\max}]$, parallelly do   INNER-GROUP   CODING ($s,\mathcal{G}^r_i, \textnormal{u},N_{\mathcal{G}^r_i}$)
			\EndFor

			\ElsIf{$s^{*}<2$ and $\lceil\frac{K}{s}\rceil\leq\alpha_{\max}$}
			   \For{$r\in[\beta_2]$}
			   \State Partition: $\{\mathcal{G}^r_1,..,\mathcal{G}^r_{\lfloor\frac{K}{s}\rfloor}\}$,  $|\mathcal{G}^r_i|=s,{i\leq \lfloor\frac{K}{s}\rfloor}$
                      
			\State For  $i\in [ \lfloor\frac{K}{s}\rfloor]$ parallelly do   INNER-GROUP  
			  CODING ($s,\mathcal{G}^r_i, \textnormal{u},N_{\mathcal{G}^r_i}$)
			\EndFor
			
						\Else~
			\State $W_{n,\mathcal{T}}^{\textnormal{u}}\!\!\!\xrightarrow{\text{split}}\!\!\big(W_{n,\mathcal{T}}^{\textnormal{u}_1},W_{n,\mathcal{T}}^{\textnormal{u}_2}\big)$,  $\frac{|W_{n,\mathcal{T}}^{\textnormal{u}_1}|}{|W_{n,\mathcal{T}}^{\textnormal{u}_2}|}\! =\! \frac{\lfloor\frac{K}{s}\rfloor(s-1)}{s^{*}-1}$
		\For{$r\in[\beta_3]$}
			\State Partition: $\{\mathcal{G}^r_1,..,\mathcal{G}^r_{\lceil\frac{K}{s}\rceil}\}$ with\\
			\qquad\qquad\qquad\qquad$|\mathcal{G}^r_i|=s,{i\leq \lfloor\frac{K}{s}\rfloor};|\mathcal{G}_{\lceil\frac{K}{s}\rceil}|=s^{*}$		
			\State \hspace{-45pt}  parallel$\left\{\begin{array}{l}
			\text{INNER-GROUP  CODING}~(s,\mathcal{G}^r_i,\textnormal{u}_1,N_{\mathcal{G}^r_{i}}),\forall i\quad\quad\quad\\
			\text{INNER-GROUP  CODING with parameters:} ~(s,\mathcal{G}^r_{\lceil\frac{K}{s}\rceil},\textnormal{u}_2,{s-1 \choose |\mathcal{G}^r_{\lceil\frac{K}{s}\rceil}| -1}N_{\mathcal{G}^r_{\lceil\frac{K}{s}\rceil}})
			\end{array}\right.$
			\EndFor
			\EndIf
			\EndFor
			\EndProcedure

				\Procedure{Inner-Group  Coding~}{$s,\mathcal{G}, \textnormal{p},\gamma$}
			\State   $W_{n,\mathcal{T}}^{\textnormal{p}}\xrightarrow{\text{split}}\big(W_{n,\mathcal{T}}^{\textnormal{p},l}:l\in[(|\mathcal{G}|-1)\gamma]\big)$  			
			\For{$  \mathcal{S}\subseteq [K]: |\mathcal{S}| = s, \mathcal{G}\subseteq\mathcal{S}  $}
			\State 
		$	X_{k,\mathcal{G},s}^{{\text{p}}} \leftarrow   \oplus_{j \in \mathcal{G}\backslash\{k\}}
			W_{d_j,\mathcal{S} \backslash\{j\}}^{{\text{p}}, l{(j,\mathcal{G},\mathcal{S})}}, \forall k\in \mathcal{G} $

			\EndFor
			\EndProcedure

		\end{algorithmic}   \end{spacing} 
	\end{algorithm}
	
	}
	

	\section{Conclusions}\label{Sec_Conclusion}
	In this paper, we considered a cache-aided broadcast network with user cooperation where users can exchange data with each other via a shared link or a flexible routing network. 
	We proposed two innovative coded caching schemes for centralized and decentralized placement respectively. Both   schemes achieve a parallel gain and a cooperation gain in terms of communication delay} by exploiting parallel transmission between the server and users and among the users themselves.  Furthermore, we showed that for in centralized  caching case, letting too many users parallelly send information could be harmful.  The information theoretic converse bounds were established and we proved that the centralized scheme achieves the optimal transmission delay within a constant multiplicative gap in all regimes, and the decentralized scheme becomes order optimal when the cache size of each user is larger than a small threshold which
	tends to zero as the number of users tends to infinity. Our work indicates that  user cooperation and coded caching both are promising techniques to reduce the transmission delay and should be jointly considered in the distributed system which suffers from data congestion problem. 

	\appendices

	\section{Proof of The Converse}\label{App_converse}

	The proof of the lower bound follows similar idea from \cite{Centralized}. Note that  due to the flexibility of  cooperation network, the connection and partitioning status between users can change during the delivery phase, we can not drive the lower bound  directly as in \cite{Centralized}.  Moreover, the parallel transmission of the server and many users results in abundant transmitting signals, making the scenario more sophisticated.  
	
	Let $T^*_1$ and $T^*_2$ denote the optimal rate sent by the server and each user. We first consider an ideal case where every user is served by a exclusive server and  user, which both store full files in the database, then  we easy to obtain $T^*\geq \frac{1}{2}(1-\frac{M}{N}).$ 
	
	Next,  consider the first $s$ users with cache contents $Z_1,...,Z_s$. Define $X_{1,0}$ to be the signal sent by the server, and $X_{{1,1}},\ldots,X_{{1,\alpha_\text{max}}}$ to be the signals sent by the $\alpha_\text{max}$ users, respectively,  where $X_{{j,i}}\in[\lfloor 2^{T^*_2F} \rfloor]$ for $j\in[s]$ and $i\in[\alpha_\text{max}]$. Assume that $W_1,\ldots,W_s$ is determined by  $X_{1,0}$, $X_{{1,1}},\ldots,X_{{1,\alpha_\text{max}}}$ and  $Z_1,\ldots,Z_s$. Also,  define $X_{2,0}$, $X_{{2,1}},\ldots,X_{{2,\alpha_\text{max}}}$ to be the signals  which enable the users to decode $W_{s+1},...,W_{2s}$. Continue the same process such that $X_{\lfloor N/s\rfloor,0}$, $X_{{\lfloor N/s\rfloor,1}},\ldots,X_{{\lfloor N/s\rfloor,\alpha_\text{max}}}$  are the signals which enable the users to decode $W_{s\lfloor N/s\rfloor-s+1},...,W_{s\lfloor N/s\rfloor}$. We then have $Z_1,\ldots,Z_s$,  $X_{1,0},\ldots,X_{\lfloor N/s\rfloor,0}$, and $$X_{{1,1}},  \ldots,X_{{1,\alpha_\text{max}}},\ldots, X_{{\lfloor N/s\rfloor,1}},\ldots,X_{{\lfloor N/s\rfloor,\alpha_\text{max}}}$$ to  determine $W_{1},\ldots,W_{s\lfloor N/s\rfloor}$. Let
	\[ 
	{\bf{X}}_{1:\alpha_\text{max}} \triangleq (X_{{1,1}},  \ldots,X_{{1,\alpha_\text{max}}},\ldots, X_{{\lfloor N/s\rfloor,1}},\ldots,X_{{\lfloor N/s\rfloor,\alpha_\text{max}}}). 
	\]
	By the definitions of $T^*_1$,  $T^*_2$ and the encoding function \eqref{eq:userencoding}, we have 
	\begin{subequations}\label{eq_entropy}
		\begin{IEEEeqnarray}{rCl}
			&&H(X_{1,0},\ldots,X_{\lfloor  {N}/{s}\rfloor,0}) \leq \lfloor  {N}/{s}\rfloor T^*_1F,\\
			&&H({\bf{X}}_{1:\alpha_\text{max}}) \leq \lfloor {N}/{s}\rfloor \alpha_\text{max} T^*_2F, \\
			&&H({\bf{X}}_{1:\alpha_\text{max}},Z_1,\ldots,Z_s) \leq  KMF.
		\end{IEEEeqnarray}
	\end{subequations}
	
	Consider then the cut separating $X_{1,0},\ldots,X_{\lfloor N/s\rfloor,0}$, ${\bf{X}}_{1:\alpha_\text{max}}$, and $Z_1,\ldots,Z_s$ from the corresponding $s$ users. By the cut-set bound and \eqref{eq_entropy}, we have
	\begin{IEEEeqnarray}{rCl}
		\lfloor \frac{N}{s}\rfloor sF &\leq&\lfloor \frac{N}{s}\rfloor T^*_1F+KMF,\\
		\lfloor \frac{N}{s}\rfloor sF   & \leq&  \lfloor \frac{N}{s}\rfloor T^*_1F+sMF+\lfloor \frac{N}{s} \rfloor \alpha_\text{max} T^*_2F.
	\end{IEEEeqnarray}
	Since we have $T^*\geq T^*_1$ and $T^*\geq \max\{T^*_1,T^*_2\}$ from the above definition, solving for $T^*$ and optimizing over all possible choices of $s$, we obtain
	\begin{subequations}\label{eq_proofCutset}
		\begin{IEEEeqnarray}{rCl}
			T^*&\geq &\max\limits_{s\in [K]}(s-\frac{KM}{\lfloor N/s\rfloor}),\\
			T^*&\geq& \max\limits_{s\in [K]}(s-\frac{sM}{\lfloor N/s\rfloor})\frac{1}{1+\alpha_\text{max}}.
		\end{IEEEeqnarray}
	\end{subequations}

	\section{Proof of Theorem \ref{Thrm_Gap}}\label{App_Gap}

	We prove that  $ T_\textnormal{central}$  is within a constant multiplicative gap of the minimum feasible delay $T^*$ for all values of $M$. To prove the result, we
	compare  them in the  following regimes.
	\begin{itemize}
		\item 
		If $ 0.6393 < t<  \lfloor  K/\alpha\rfloor-1$, from Theorem \ref{Thrm_LowerBound}, we have 
		\begin{equation}\label{eq_as}
		\begin{aligned}
		T^*&\geq(s-\frac{Ms}{\lfloor N/s\rfloor})\frac{1}{1+\alpha_\text{max}}\\
		&\overset{(a)}\geq \frac{1}{12}\cdot K\Big(1-\frac{M}{N}\Big)\frac{1}{1+t}\cdot\frac{1}{1+\alpha_\text{max}},
		\end{aligned}
		\end{equation}
		where (a) follows from   \cite[Theorem 3]{Centralized}. Then we have
		\begin{IEEEeqnarray}{rCl}
			\frac{ T_\textnormal{central}}{T^*}
			&\leq& 12\cdot\frac{(1+\alpha_\text{max})(1+t)}{1+t+\alpha t}\nonumber\\
			&=&12\cdot\frac{(1+\alpha_\text{max})}{1+\alpha t/(1+t)}\nonumber\\
			&\leq& 12\cdot\frac{(1+\alpha_\text{max})}{1+\alpha\cdot 0.6393/(1+0.6393)}\nonumber\\
			&\leq &31,
		\end{IEEEeqnarray}
		where the last inequality holds since we can choose $\alpha=\alpha_\text{max}$. 
		\item If $t> \lfloor  K/\alpha\rfloor-1  $, we have 
		\begin{IEEEeqnarray}{rCl}
			\frac{ T_\textnormal{central}}{T^*} &\leq  &\frac{K(1-\frac{M}{N})\frac{1}{1+t+\alpha (\lfloor K/\alpha \rfloor-1)}}{\frac{1}{2}(1-\frac{M}{N})}\nonumber\\
			&= &\frac{2K}{1+t+\alpha (\lfloor K/\alpha \rfloor-1)}\nonumber\\
			&\overset{(a)}\leq  &\frac{2K}{K+KM/N}\nonumber\\
			&\leq& 2,
		\end{IEEEeqnarray}
		where $(a)$ follows from that  we can  choose $\alpha=1$.

		\item If $t\leq 0.6393$, 
		setting $s=0.275N$, we have
		\begin{IEEEeqnarray}{rCl}\label{eq_lower}
			T^* &\geq& s-\frac{KM}{\lfloor N/s\rfloor}\nonumber\\
			&\overset{(a)}\geq& s-\frac{KM}{N/s-1}\nonumber\\
			&=&0.275N-t\cdot0.3793N\nonumber\\
			&\geq& 0.0325N> \frac{1}{31}\cdot N,
		\end{IEEEeqnarray}
		
		where $(a)$ holds since $\lfloor x \rfloor \geq x-1$ for any $x \geq 1$. 
		Note that for all values of $M$, the transmission delay 
		\begin{equation}\label{upper}
		 T_\textnormal{central}\leq \min \{K, N\}. 
		\end{equation}
		Combining with \eqref{eq_lower} and \eqref{upper}, we have $$\frac{{ T_\textnormal{central}}}{{T^*}} \leq 31.$$

	\end{itemize}

	\section{Proof of Corollary \ref{Coro_UpperBound_Decen}}\label{App_Bound_Decen}
	The function $R_\text{u}$ has three distinct forms for different values of $\alpha_{\max}$. Thus, we discuss $R_\text{u}$ in three regimes of $\alpha_{\max}$: $\alpha_{\max} = \lfloor\frac{K}{2} \rfloor$, $\alpha_{\max}=1$ and $1<\alpha_{\max}<\lfloor\frac{K}{2} \rfloor$ respectively. For convenience, we define $q = 1-p$.
	\subsection{$ \alpha_{\textnormal{max}} = \lfloor \frac{K}{2} \rfloor$}
	When  $\alpha_{\max} = \lfloor\frac{K}{2} \rfloor$, we have
	\begin{IEEEeqnarray}{rCl}
		R_\textnormal{u}&=&R_\textnormal{u-f} \nonumber\\
		&\triangleq& \sum_{s=2}^{K}\frac{K{K-1 \choose s-1}}{f(K,s)}p^{s-1}q^{K-s+1},
	\end{IEEEeqnarray}
where ${R}_{\textnormal{u-f}}$ denotes the user's transmission rate for a fully flexible cooperation network  with  $\alpha_{\textnormal{max}}=\lfloor \frac{K}{2} \rfloor$. In the fully flexible cooperation network,  at most $\lfloor \frac{K}{2} \rfloor$ users are allowed to transmit messages simultaneously, in which   the user transmission turns to unicast.
	Note that in each term of the summation:
	\begin{IEEEeqnarray}{rCl}
		\frac{K{K-1 \choose s-1}}{f(K,s)}&& \leq \frac{K{K-1 \choose s-1}}{K-1- \frac{K}{s}}\nonumber \\
		&& = \Big(\frac{K}{K-1}+\frac{\big(\frac{K}{K-1}\big)^2}{s-\frac{K}{K-1}} \Big) \cdot {K-1 \choose s-1} \label{eq:pf_ex_up1}\nonumber \\
		&&\leq\frac{K{K-1 \choose s-1}}{K-1} + \frac{2K{K \choose s}}{(K-1)(K-2)},
		\label{eq:pf_ex_up2}
	\end{IEEEeqnarray}
	where the last inequality  holds by $s\geq \frac{K}{K-1}+\frac{K-2}{K-1}=2$ and 
	\begin{IEEEeqnarray}{rCl}
		\frac{\big(\frac{K}{K-1}\big)^2}{s-\frac{K}{K-1}}{K-1 \choose s-1}
		&& = \frac{K^2{K-1 \choose s-1}}{(K-1)(K-2)}\cdot \frac{\frac{K-2}{K-1}}{s-\frac{K}{K-1}} \nonumber\\
		&& \leq \frac{K^2{K-1 \choose s-1}}{(K-1)(K-2)} \cdot \frac{\frac{K-2}{K-1}+\frac{K}{K-1}}{s-\frac{K}{K-1}+\frac{K}{K-1}} \nonumber \\
		&& = \frac{2K}{(K-1)(K-2)} \cdot {K \choose s}. \nonumber
	\end{IEEEeqnarray}
	Therefore, by   (\ref{eq:pf_ex_up2}), $R_\text{u-f}$ can be rewritten as
	\begin{IEEEeqnarray}{rCl}
		\hspace{0pt}R_\textnormal{u-f} && \leq \frac{K}{K-1}\sum_{s=2}^K{K-1 \choose s-1}p^{s-1}q^{K-s+1}+\nonumber\\ && 
		\hspace{10pt}
		\frac{2K}{(K-1)(K-2)}\sum_{s=2}^K{K \choose s}p^{s-1}q^{K-s+1} \nonumber \\
		&& \hspace{-10pt} \overset{i \triangleq s-1}{=} \frac{Kq}{K-1}\cdot\sum_{i=1}^{K-1}{K-1 \choose i}p^iq^{K-1-i}+ \nonumber \\ &&\hspace{10pt} \frac{2Kq/p}{(K-1)(K-2)}\cdot\sum_{s=2}^{K}{K\choose s}p^sq^{K-s} \nonumber\\
		&& = \frac{Kq}{K-1}\Big(1-q^{K-1}\Big) + \frac{2Kq/p}{(K-1)(K-2)} \nonumber \\ && \hspace{85pt}\cdot\Big(1-q^K-Kpq^{K-1}\Big). \nonumber
	\end{IEEEeqnarray}
	


	\section{Proof of Theorem \ref{Thrm_Gap_Decen}}\label{App_Gap_Decen}
	Before giving the proof of Theorem \ref{Thrm_Gap_Decen}, we first introduce the following two lemmas.
	\begin{Lemma}\label{Lem:secant}
		Given arbitrary convex function $g_1(p)$ and arbitrary concave function $g_2(p)$, if they  intersect at two points with $p_1 < p_2$, then $g_1(p) \leq g_2(p)$ for all $p \in [p_1,p_2]$.
	\end{Lemma}
	
	\begin{Lemma}\label{Coro_Threshold_Decen}
		For memory size $0\leq p \leq 1$ and maximum number of allowed users $1 \leq \alpha_\textnormal{max} \leq \lfloor \frac{K}{2} \rfloor$, we have
		\[R_\textnormal{u} \geq R_\emptyset, \quad\text{for~all~} p\in[p_\textnormal{th} ,1].\]
	\end{Lemma}
	\begin{proof}
		When $\alpha_{\max} =\lfloor\frac{K}{2} \rfloor$, from Equation (\ref{eq:Ru}), we have
		\begin{subequations}
			\begin{IEEEeqnarray}{rCl}
				R_\textnormal{u}|_{\alpha_{\max}=\lfloor\frac{K}{2} \rfloor}  &&= \sum_{s=2}^K \frac{K{K-1 \choose s-1}}{f(K,s)} p^{s-1}(1-p)^{K-s+1} \label{eq:LB1} \\
				&&\hspace{-9pt} \overset{x \triangleq s-1}{\geq} \frac{K}{K}\sum_{x=1}^{K-1}{K-1\choose x}p^{x}(1-p)^{K-x} \label{eq:LB2}\\
				&& = \big(1-p\big) \cdot \big( 1-(1-p)^{K-1}\big) \triangleq \ubar{R}_\textnormal{u-f},\quad\label{eq:LB3}
			\end{IEEEeqnarray}
		\end{subequations}
		where (\ref{eq:LB1}) to (\ref{eq:LB2}) utilize the fact that $\frac{K}{K-1-\lfloor\frac{K}{s}\rfloor} > \frac{K}{K-1}$. Thus,  \[R_\textnormal{u}\geq R_\textnormal{u}|_{\alpha_{\max}=\lfloor\frac{K}{2} \rfloor}\geq \ubar{R}_\textnormal{u-f},\]
		We can rewrite $R_\emptyset(p)$ as 		
		 $$R_\emptyset(p) = K(1-p)^K.$$ 
		Since $\frac{\partial^2 \ubar{R}_{\textnormal{u-f}}(p)}{\partial p^2} <0$ and $\frac{\partial^2 R_\emptyset(p)}{\partial p^2}>0$, $\ubar{R}_{\textnormal{u-f}}(p)$ is a concave function while $R_\emptyset(p)$ is a convex function, and they intersect  at $p_1 = p_\textnormal{th} = 1 - \big(\frac{1}{K+1}\big)^\frac{1}{K-1}$ and $p_2 = 1$ while $p_\textnormal{th} \leq 1$. Therefore, by Lemma \ref{Lem:secant},  for all  $p \in [p_\textnormal{th},1]$, we have
		\[R_\textnormal{u}\geq R_\textnormal{u-f} \geq \ubar{R}_\textnormal{u-f} \geq R_\emptyset.\]
	\end{proof}

	Define $T_\textnormal{decentral}^{'} \triangleq \frac{R_\textnormal{s}R_\textnormal{u}}{R_\textnormal{s}+R_\textnormal{u}-R_\emptyset}$, which can be written in another form as 
	\begin{IEEEeqnarray}{rCl}
		T_\textnormal{decentral}^{'} = R_\emptyset + \frac{(R_\textnormal{s}-R_\emptyset)(R_\textnormal{u}-R_\emptyset)}{R_\textnormal{u}+R_\textnormal{s}-R_\emptyset}. \label{eq:Prf_Thrm_Gap}
	\end{IEEEeqnarray}
	If $R_\textnormal{u} \leq R_\emptyset$, then $T_\textnormal{decentral}^{'} \leq R_\emptyset$, otherwise $T_\textnormal{decentral}^{'} \geq R_\emptyset$. 
	
	Let $q = 1-p$ for convenience. Eq. (\ref{eq:Prf_Thrm_Gap}) indicates that the value of $T_\textnormal{decentral}/T^*$ is divided into bounded and unbounded regions:
	\begin{itemize}
		\item If $p \geq p_\textnormal{th}$, then 
		\begin{IEEEeqnarray}{rCl}\label{Remark2Eq}
		\frac{T_\textnormal{decentral}}{T^*} = \frac{T_\textnormal{decentral}^{'}|_{p\geq p_\textnormal{th}}}{T^*} \leq \frac{T_\textnormal{decentral}^{'}}{T^*},
\end{IEEEeqnarray}
	 which is bounded by a constant.
		\item If $p < p_\textnormal{th}$, it is not sure whether $R_\textnormal{u}$ is greater than $R_\emptyset$ or not. Therefore, 
		\begin{IEEEeqnarray}{rCl}
			\frac{T_\textnormal{decentral}}{T^* }&& = \max\{\frac{R_\emptyset|_{p<p_\textnormal{th}}}{T^*}, \frac{T_\textnormal{decentral}^{'}|_{p<p_\textnormal{th}}}{T^*}\} \nonumber \\
			&& \leq \max\{\frac{R_\emptyset}{T^*}, \frac{T_\textnormal{decentral}^{'}}{T^*}\},
		\end{IEEEeqnarray}
which might be  unbounded except for $\alpha_\text{max}=1$.
	\end{itemize}
	Now we inspect the following different situations.

	\subsection{Bounded Region for $\alpha_\textnormal{max} = \lfloor \frac{K}{2} \rfloor$ When $p \geq p_\textnormal{th}$}
	From   (\ref{eq: Flex_UpperBound2}) and (\ref{eq:Rs}),
	\begin{IEEEeqnarray}{rCl}
		\bar{R}_\textnormal{u-f} & =& \frac{K}{K-1}\cdot \Big(q-q^K \Big) + \frac{2K}{(K-1)(K-2)} \nonumber \\
		&& \hspace{80pt} \cdot \frac{q}{p}\Big(1-q^K-Kpq^{K-1} \Big)\label{eq:pf_Flex_Upper1} \\
		& \overset{(a)}{\leq}& \frac{K}{K-1}\cdot \Big(q-q^K \Big) + \frac{2K}{(K-1)(K-2)} \nonumber \\
		&&\hspace{50pt}\cdot \frac{q}{p}\Big(1-\big(1-Kp \big)-Kpq^{K-1} \Big)
		\\
		& =& \frac{K(3K-2)}{(K-1)(K-2)}\cdot\Big(q-q^K \Big),\label{eq:pf_Flex_Upper2}\\
		R_\textnormal{s} & =& \frac{q}{p}\Big(1-q^K\Big) \overset{(b)}{\leq} \frac{q}{p}\Big(1-\big(1-Kp\big)\Big) = Kq, \label{eq:pf_srvr_Upper}
	\end{IEEEeqnarray}
	where $(a)$ and $(b)$ both follow from inequality
	\begin{IEEEeqnarray}{rCl}\label{eq:ineq}
		\big(1-p\big)^K \geq \big(1-Kp\big).
	\end{IEEEeqnarray}
	
	Then, by Remark \ref{Coro_Monotonicity} and (\ref{eq:pf_Flex_Upper2}), (\ref{eq:pf_srvr_Upper}) and (\ref{eq:Re}),
	\begin{IEEEeqnarray}{rCl}
		&&T_\textnormal{decentral}^{'}|_{\alpha_{\max}= \lfloor\frac{K}{2}\rfloor}  \nonumber \\
		& & \hspace{10pt} \leq\frac{Kq\cdot \frac{K(3K-2)}{(K-1)(K-2)}\big(q-q^K\big)}{Kq+\frac{K(3K-2)}{(K-1)(K-2)}\big(q-q^K\big)-Kq^K} \nonumber \\
		& & \hspace{10pt} = \Big(3-\frac{2}{K} \Big)\cdot q. \label{eq:up_Rdp}
	\end{IEEEeqnarray}
	
By Lemma \ref{Coro_Threshold_Decen}, $T^*\geq\frac{1}{2}q$. Combine it with (\ref{eq:up_Rdp}) yields
	\[ \frac{T_\textnormal{decentral}}{T^*}\leq \frac{T_\textnormal{decentral}^{'}|_{\alpha_{\max}}}{T^*} \leq 6 - \frac{4}{K} < 6.\]
	
	\subsection{Bounded Gap for $\alpha_\textnormal{max} =1$}
	From Lemma \ref{Coro_Threshold_Decen}, eq. (\ref{eq: Shrd_UpperBound}) and (\ref{eq:ineq}),
	\begin{IEEEeqnarray}{rCl}
		\bar{R}_\textnormal{u-s} & =& \frac{q}{p} \Big( 1-\frac{5}{2}Kpq^{K-1}-4q^K+ \frac{3(1-q^{K+1})}{(K+1)p} \Big) \nonumber\\
		& \leq &\frac{q}{p}\Big( 1-\frac{5}{2}Kpq^{K-1}-4q^K+ \frac{3(K+1)p}{(K+1)p} \Big) \nonumber \\
		& =& \frac{q}{p}\Big(4\cdot(1-q^K)-\frac{5}{2}Kpq^{K-1} \Big) \nonumber\\
		& < &4R_\textnormal{s}, \label{eq:pf_shrd_upper}\\
		T^* &\geq& \max_{s \in [K]}\Big(s - \frac{sM}{\lfloor N/ s \rfloor} \Big)\frac{1}{1+\alpha_{\max}}. \label{eq:pf_shrd_lower}
	\end{IEEEeqnarray}
	
	Recalling the results in \cite[Appendix B]{Decentralized}, we have
	$$\frac{R_\textnormal{s}}{\max_{s \in [K]}\big(s - \frac{sM}{\lfloor N/ s \rfloor} \big)} \leq 12.$$
	From this result and combining  (\ref{eq:pf_shrd_upper}) and (\ref{eq:pf_shrd_lower}) together, we have
	\begin{IEEEeqnarray}{rCl}
		\frac{R_\textnormal{s}}{T^*}  \leq 12(1+\alpha_{\max}),~
		\frac{\bar{R}_\textnormal{u-s}}{T^*}  \leq 48(1+\alpha_{\max}). \quad \label{eq:pf_semi_up}
	\end{IEEEeqnarray}
	When $p \leq p_\textnormal{th}$, by  Remark \ref{Coro_Monotonicity} and \eqref{eq:pf_semi_up}, we have 
	\begin{IEEEeqnarray}{rCl}
	\frac{T_\textnormal{decentral}}{T^*} &\leq &\frac{R_\textnormal{s}}{T^*}\leq  12(1+\alpha_{\max})\leq 24.
\end{IEEEeqnarray}
	 When $p \geq p_\textnormal{th}$,	by Remark \ref{Coro_Monotonicity}, and since $T_\textnormal{decentral}^{'}|_{\alpha_{\max}=1} \leq\min\{\bar{R}_\textnormal{u-s},R_\textnormal{s}\}$, 
	\begin{IEEEeqnarray}{rCl}
\frac{T_\textnormal{decentral}}{T^*}  &= &  \frac{T_\textnormal{decentral}^{'}|_{\alpha_{\max}=1}}{T^*} \nonumber\\
&\leq & \frac{\min\{\bar{R}_\textnormal{u-s},R_\textnormal{s}\}}{T^*}\nonumber\\
		&\leq&  \min\{12(1+\alpha_{\max}),48(1+\alpha_{\max})\} \nonumber\\
		&=&  24.
	\end{IEEEeqnarray}
	
	\subsection{Bounded Gap for $1<\alpha_\textnormal{max} < \lfloor \frac{K}{2} \rfloor$ When $p \geq p_\textnormal{th}$}\label{SubSec: Semi-Flex}
	From (\ref{Semi-Flex-UpperBound}), (\ref{eq:pf_Flex_Upper2}), (\ref{eq:pf_shrd_upper}) and (\ref{eq:pf_semi_up}), when $K \geq 3$,
	\begin{IEEEeqnarray}{rCl}
		\frac{\bar{R}_\textnormal{u}}{T^*} & \leq& \frac{1}{\alpha_{\max}} \frac{\bar{R}_\textnormal{u-s}}{T^* }+  \frac{\bar{R}_\textnormal{u-f}}{T^*} \nonumber \\
		& \leq& 48\cdot\frac{1+\alpha_{\max}}{\alpha_{\max}} + \frac{2K(3K-2)(1-q^{K-1})}{(K-1)(K-2)} \label{eq:pf_semi_up2}\\
		& \leq& 48\cdot\frac{1+\alpha_{\max}}{\alpha_{\max}} + 21,
	\end{IEEEeqnarray}
	where the second term in (\ref{eq:pf_semi_up2}) follows from the fact that $$T^* \geq \frac{1}{2}q.$$ 
	Moreover, we have
	$$R_\textnormal{s}/T^* \leq 12(1+\alpha_{\max})$$
	from \eqref{eq:pf_semi_up}. In conclusion, when $K \geq 3$ and $p \geq p_\textnormal{th}$, again by   Remark \ref{Coro_Monotonicity},
	\begin{IEEEeqnarray}{rCl}
	\frac{T_\textnormal{decentral}}{T^*}  &= &  \frac{T_\textnormal{decentral}^{'}}{T^*} \nonumber\\
		&\leq&
		\min\{\frac{\bar{R}_\textnormal{u}}{T^*}  , \frac{R_\textnormal{s}}{T^*} \} \nonumber\\
		&\leq& \min_{\alpha_{\max}}\{48\cdot\frac{1+\alpha_{\max}}{\alpha_{\max}}\! +\! 21, 12(1\!+\!\alpha_{\max})\}\nonumber\\
		&=& 77. 
	\end{IEEEeqnarray}
	
	When $K<3$, by Inequality (\ref{eq:pf_srvr_Upper}) and Remark \ref{Coro_Monotonicity},
	\[ \frac{T_\textnormal{decentral}'}{T^*} \leq \frac{R_\textnormal{s}}{T^*} \leq \frac{Kq}{\frac{1}{2}q} \leq 4 < 77.\]
	
	Hence, for all values of $K$, we have $$\frac{T_\textnormal{decentral}'}{T^*} \leq \max\{4,77\} = 77.$$

	\subsection{Unbounded Gap}
	For the unbounded region, maybe $R_\textnormal{u} > R_\emptyset$ and $T_\textnormal{decentral} =T_\textnormal{decentral}^{'}$, or maybe $T_\textnormal{decentral} = R_\emptyset$. The multiplicative gap corresponding to the first case was already calculated in previous subsections, while the gaps for the second case is derived in the following:
	
	First, according to Lemma \ref{Coro_Threshold_Decen}, $T^*$ has two different lower bounds: $T^*  \geq \frac{1}{2}q \triangleq  R_1^{*}$, and 
	\begin{IEEEeqnarray}{rCl}
		T^* && \geq \max_{s \in [K]}\Big(s-\frac{KM}{\lfloor N/s \rfloor} \Big)  \geq \max_{s \in [K]}\Big(s-\frac{KM}{N/(2s)} \Big)\triangleq R_2^{*}.
		\nonumber 
	\end{IEEEeqnarray}
	The quotients of $R_\emptyset$ divided by those lower bounds changes monotonically,
	\begin{IEEEeqnarray}{rCl}
		\frac{\partial \big(R_\emptyset/R_1^{*}\big)}{\partial p} && = \frac{\partial \big(2K(1-p)^{K-1}\big)}{\partial p} \nonumber \\ &&\leq 0,\nonumber \\
		\frac{\partial \big(R_\emptyset/R_2^{*}\big)}{\partial p} && = \frac{\partial \big(q^K/(1-2Kp)\big)}{\partial p} \nonumber \\
		&& = \frac{Kq^{K-1}\big(1+2(K-1)p\big)}{(1-2Kp)^2} \nonumber \\&& \geq 0. \nonumber
	\end{IEEEeqnarray}
	Also notice that when $p=0$, $\big(R_\emptyset/R_2^{*}\big) = 1 < R_\emptyset/R_1^{*}$; while if $p=1$, $R_\emptyset/R_2^{*}  > R_\emptyset/R_1^{*} = 1$. Therefore, the maximum of \[R_\emptyset/\max\{R_1^{*},R_2^{*}\}\] exists at the intersection between $R_1^{*}(p)$ and $R_2^{*}(p)$, where $p^{*} \triangleq \frac{1}{2K+1}: R_1^{*}(p^{*}) = R_2^{*}(p^{*})$. Therefore,
	\[\frac{R_\emptyset}{T^*} \leq \frac{R_\emptyset(p^{*})}{R_1^{*}(p^{*})} = 2K\Big(\frac{2K}{2K+1}\Big)^{K-1}.\]
	
	Next, by the definition that $R_\emptyset \leq R_s$, and eq. (\ref{eq:pf_semi_up}),
	\[R_\emptyset/T^* \leq R_\textnormal{s}/T^* = 12(1+\alpha_{\max}).\]
	
	Finally, $R_\emptyset/T^*$ is smaller than
	\[\min\{2K\Big(\frac{2K}{2K\!+\!1}\Big)^{K\!-\!1}, 12(1\!+\!\alpha_{\max}) \}.\]

\end{document}